\documentclass[%
notitlepage,
superscriptaddress,
preprintnumbers,
amsmath,amssymb,
aps,
]{revtex4-1}

\usepackage{graphicx}
\usepackage{dcolumn}
\usepackage{bm}

\usepackage{amsmath,amsfonts,amssymb,amsthm}
\usepackage{graphicx}
\usepackage{comment}
\usepackage[colorlinks=true, allcolors=blue]{hyperref}
\usepackage[thinc]{esdiff}
\usepackage{enumitem}

\usepackage[T1]{fontenc}
\usepackage{braket}
\usepackage{graphicx}
\usepackage{hyperref}
\usepackage{appendix}
\usepackage{subcaption}
\usepackage{tikz}

\newcommand{\wn}{\nu}
\newcommand{\et}[1]{e^{\imath #1 \theta_0}}
\newcommand{\etm}[1]{e^{-\imath #1 \theta_0}}
\newcommand{\coeff}[1]{\alpha_{#1}^{(n)}}
\newcommand{\xw}{X^{(n)}}
\newcommand{\uw}{U^{(n)}}
\newcommand{\vw}{V^{(n)}}
\newcommand{\yw}{Y^{(n)}}
\newcommand{\amm}{a_-^{-1}}
\newcommand{\amp}{a_-}
\newcommand{\apm}{a_+^{-1}}
\newcommand{\app}{a_+}
\newcommand{\ucoeff}[1]{\alpha_{#1}^{(n)}}
\newcommand{\ts}{(<)}
\newcommand{\tl}{(>)}

\newcommand{\norm}[1]{\left\lVert#1\right\rVert}
\newcommand{\normb}[1]{\left\lVert#1\right\rVert_1}

\newtheorem{theorem}{Theorem}
\newtheorem{lemma}{Lemma}

\numberwithin{equation}{section}

\begin{document}
	
\title{Asymptotic behavior of Toeplitz determinants with {a delta function singularity}}

\author{Vanja Mari\'{c}}
\affiliation{Division of Theoretical Physics, Ru\dj{}er Bo\v{s}kovi\'{c} Institute, Bijeni\u{c}ka cesta 54, 10000 Zagreb, Croatia}
\affiliation{SISSA and INFN, via Bonomea 265, 34136 Trieste, Italy.}

\author{Fabio Franchini}
\affiliation{Division of Theoretical Physics, Ru\dj{}er Bo\v{s}kovi\'{c} Institute, Bijeni\u{c}ka cesta 54, 10000 Zagreb, Croatia}

\begin{abstract}
We find the asymptotic behaviors of Toeplitz determinants with symbols which are a sum of two contributions: one analytical and non-zero function in an annulus around the unit circle, and the other proportional to a Dirac delta function. The formulas are found by using the Wiener-Hopf procedure. The determinants of this type are found in computing the spin-correlation functions in low-lying excited states of some integrable models, where the delta function represents a peak at the momentum of the excitation. As a concrete example of applications of our results, using the derived asymptotic formulas we compute the spin-correlation functions in the lowest energy band
of the frustrated quantum XY chain in zero field, and the ground state magnetization.
\end{abstract}

\preprint{RBI-ThPhys-2020-16}

\maketitle

\section{Introduction}\label{section Introduction}
We consider Toeplitz determinants
\begin{equation}\label{first}
\tilde{D}_n(\tilde{f})=\det \big(\tilde{f}^{(n)}_{j-k}\big)_{j,k=1}^{n} \; , \qquad \tilde{f}_j^{(n)}=\frac{1}{2\pi} \int_{0}^{2\pi} \tilde{f}(\theta,n) e^{-\imath j \theta} d \theta
\end{equation}
with a symbol 
\begin{equation}\label{symbol}
\tilde{f}(\theta,n)=f(e^{\imath \theta} ) \big[1+{2\pi}z_n \ \delta(\theta-\theta_0)\big]
\end{equation}
Here $\delta$ is Dirac delta function, $\theta_0\in[0,2\pi)$, $(z_n)_{n\in\mathbb{N}}$ is an arbitrary sequence in $\mathbb{C}$ and $f$ is a continuous function on the unit circle.

It follows that for $\theta_0\neq0$ the elements {$\tilde{f}_j^{(n)}$} are equal to
\begin{equation}\label{elements tilde}
\tilde{f}_j^{(n)}=f_j+z_n f(e^{\imath \theta_0}) e^{-\imath j \theta_0},
\end{equation}
where
\begin{equation}
f_j=\frac{1}{2\pi}\int_0^{2\pi}f(e^{\imath \theta})e^{-\imath j \theta} d\theta \; .
\end{equation}
For $\theta_0=0$ there is an ambiguity in the delta function integral. In this case we use \eqref{elements tilde} to define the coefficients {$\tilde{f}_j^{(n)}$} and the Toeplitz matrix $\tilde{D}_n(\tilde{f})$.

We restrict ourselves to symbols $f$ that are non-zero and analytic in an annulus including the unit circle. A general such  symbol can be written as
\begin{equation}\label{symbol f}
f(z)=a(z)z^\wn,
\end{equation}
where $a$ is a function that is analytic and non-zero on the annulus and has zero winding number, while $\wn\in\mathbb{Z}$ is the winding number of the symbol $f$ {(see e.g. \cite{Bottcher1999, Rao1991})}.

We are interested in asymptotic formulas for $\tilde{D}_n(\tilde{f})$ {as $n\to\infty$}. The asymptotic formulas for
\begin{equation}
D_n(f)=\det (f_{j-k})_{j,k=1}^{n}
\end{equation}
for analytic non-zero symbols $f$ are by now considered classical, and exist under much more general conditions. For $\wn=0$ they are given by the strong Szeg\H{o} limit theorem (originally proven in \cite{Szego}, for a review of later developments see \cite{DeiftItsKrasovsky,Simon_orthogonal,Bottcher_book}). The asymptotic formulas for nonzero $\wn$ have been first obtained in \cite{Hartwig1969,FisherHartwig}, and later under different conditions and using different methods in \cite{BottcherSilbermann,BottcherWidom}. The delta function in the symbol \eqref{symbol} might also be considered, in a suitable limit, as a singularity of the symbol, different from the widely studied Fisher-Hartwig singularities (for a review see \cite{Krasovsky_aspects}).

We decided to solve this problem motivated by the appearance of determinants of type \eqref{first} in spin-correlation functions of certain low energy states in quantum spin chains mappable to free fermonic systems. In such instances, the determinant $D_n(f)$ reflects the ground state correlations (in absence of frustration) and the delta function of the symbol has a peak at the momentum of the fermionic excitation on the vacuum state. For chains with boundary frustration, the spin-correlation functions in the lowest admissible state are already determined by \eqref{first}, where $\theta_0$ emerges as the mode minimizing the energy and lying at the bottom of a band of states (in the thermodynamic limit) \cite{DongPRE,Dong_2016,mari2019frustration,mari2020frustration,Giampaolo:2018imc}. The leading asymptotic term for particular determinants of the type \eqref{first} in the case $\wn=0$ was found in \cite{Dong_2016} in the context of the frustrated quantum Ising chain, without discussion of the subleading terms and without rigor: providing a reliable proof has been the initial inspiration for this work, together with the possibility of extending the conditions of applicability.

To introduce the notation,
we are first going to review some results on the asymptotics of the determinant $D_n(f)$ where $f$ is non-zero and analytic in an annulus around the unit circle. The asymptotic formulas of \cite{FisherHartwig,Hartwig1969} are appropriate for this case. The function $a(z)=\sum_{k=-\infty}^{\infty}a_k z^k$, defined in \eqref{symbol f}, is analytic in an annulus including the unit circle so
\begin{equation}\label{rho plus minus}
\limsup_{k\to\infty}|a_{-k}|^{1/k}=\rho_{-}<1<\rho_+=\liminf_{k\to\infty}|a_k|^{-1/k} \; .
\end{equation}
{An analytic logarithm of $a$ on $\rho_-<|z|<\rho_+$ exists so we can introduce the Wiener-Hopf factorization (see e.g. \cite{Bottcher1999,TheTwo-DimensionalIsingModel})}
\begin{equation}\label{apmdef}
\amp(z)=\exp \sum_{k=1}^{\infty} (\log a)_{-k}z^{-k}  , \quad \app(z)=\exp \sum_{k=0}^{\infty} (\log a)_{k}z^{k} ,
\end{equation}
where $\log a (z)= \sum_{k=-\infty}^{\infty} (\log a)_{k} z^{ k}$ and thus $a=a_-a_+$. We also introduce the functions
\begin{equation}
b=a_-a_+^{-1} , \quad c=a_+a_{-}^{-1}, \qquad b \ c =1 \; .
\end{equation}
For expressing the subleading terms in the asymptotic formulas for $D_n(f)$ it is useful to define $\rho\in(0,1)$ by
\begin{equation}\label{rho definition}
\rho_{-}<\rho<1<\rho^{-1}<\rho_+ \; .
\end{equation}
The function $a$ is then analytic on $\rho\leq|z|\leq\rho^{-1}$ and $a_j=O(\rho^{j}), \; a_{-j}=O(\rho^j)$, for $j\geq 0$, for all $\rho$ satisfying \eqref{rho definition}. Analogous relations hold for the functions $b$ and $c$.

The asymptotic behavior of $D_n(f)$ in the case of zero winding number of the symbol ($\wn=0$) is given by the strong Szeg\H{o} limit theorem. The version of \cite{FisherHartwig,Hartwig1969} in the case of analytic symbol reads
\begin{equation}\label{FH zero}
D_n(a)=\exp\big[n(\log a)_0 +\sum_{k=1}^{\infty}k (\log a)_k (\log a)_{-k} +O(\rho^{2n})\big] \quad {\textrm{as $n\to\infty$.}}
\end{equation}
We note that in the same reference an explicit formula for the term $O(\rho^{2n})$ in \eqref{FH zero} is given, up to corrections $O(\rho^{4n})$.

For $\wn\neq 0$ the asymptotic behavior of $D_n(f)$ is determined by the asymptotic behavior of $D_n(a)$ and the determinant of the $|\wn|\times|\wn|$ Toeplitz matrix, defined by
\begin{equation}\label{Delta}
\Delta_{\wn,n}=\det\big(d^{(n)}_{j-k}\big)
_{j,k=1}^{|\wn|} \ ,\quad \textrm{where } \quad d_j^{(n)}=\begin{cases} b_{n+j} &\textrm{for } \wn <0\\
c_{-n-j} &\textrm{for } \wn >0
\end{cases} .
\end{equation}
The asymptotic formula is
\begin{equation}\label{FH nonzero}
D_n\big(f\big)=(-1)^{n\wn}D_{n+|\wn|}(a)\big[\Delta_{\wn,n}+O(\rho^{n(|\wn|+3)}) \big] .
\end{equation}
The formula \eqref{FH nonzero} follows from the more precise result of \cite{FisherHartwig,Hartwig1969}{, given by Theorem 4 in \cite{Hartwig1969} and Theorem 6 in \cite{FisherHartwig}.}

We will extend these formulas to determinants of type \eqref{first}. In order to do so, we first define the determinants $\tilde{\Delta}_{\wn,n}(l)$, for $l=1,2,...,|\wn|$, for $\wn\neq 0$, as the determinants of the matrix {$(d_{j-k}^{(n)})_{j,k=1}^{|\wn|}$} with the column $l$ replaced by the vector $(1,e^{-\imath\theta_0},e^{-\imath 2\theta_0},...,e^{-\imath (|\wn|-1) \theta_0})^{\textrm{T}}$ for $\wn<0$ and by $(1,e^{\imath\theta_0},e^{\imath 2\theta_0},...,e^{\imath (\wn-1) \theta_0})^{\textrm{T}}$ for $\wn>0$. This definition can be written as
\begin{equation}\label{Delta tilde l}
\tilde{\Delta}_{\wn,n}(l)=\det\big(\tilde{d}^{(l)}_{j,k}\big)_{j,k=1}^{|\wn|} , \quad \textrm{where }
\tilde{d}_{j,k}^{(l)}=
\begin{cases}
(1-\delta_{k,l}) { d_{j-k}^{(n)} } +\delta_{k,l}e^{-\imath (j-1)\theta_0} &\textrm{for } \wn<0 \\ 
(1-\delta_{k,l}) {d_{j-k}^{(n)}  } +\delta_{k,l}e^{\imath (j-1)\theta_0} &\textrm{for } \wn>0 
\end{cases}.
\end{equation}

The main results of this work are the following two theorems.
\begin{theorem}\label{tm zero winding}
	For $\wn=0$ we have
	\begin{equation}
\tilde{D}_n(\tilde{f})=D_n(a)\bigg\{1+z_n\bigg[ n+ \imath\frac{d}{d \theta}\log b(e^{\imath \theta}) \bigg|_{\theta=\theta_0}+O(\rho^n)  \bigg] \bigg\}    \quad \textrm{as } n\to \infty,
	\end{equation}
	where $\rho$ is defined by \eqref{rho definition}.
\end{theorem}

\begin{theorem}\label{tm non-zero winding}
	Let $\wn\neq 0$ and suppose $D_n(f)\neq0$, $\Delta_{\wn,n}\neq 0$, for $n\geq n_0$, $n_0\in\mathbb{N}$. If for sufficiently small $\rho$, satisfying \eqref{rho definition}, we have
	\begin{equation}\label{condition tm}
	\lim_{n\to\infty}\frac{\tilde{\Delta}_{\wn,n}(j)}{\Delta_{\wn,n}}\rho^{2n}=0  \quad \textrm{for all }j\in\{1,2,...,|\wn|\}\;,
	\end{equation}
	where $\tilde{\Delta}_{\wn,n}(j)$ are defined by \eqref{Delta tilde l},
	then for $n\geq n_0$
	\begin{align}
	&\tilde{D}_n(\tilde{f})=D_n(f)\bigg\{1+z_n\bigg[-b(e^{\imath \theta_0}) \etm{(n+1)}\sum_{j=1}^{|\wn|} \frac{\tilde{\Delta}_{\wn,n}(j)}{\Delta_{\wn,n}}\Big(\et{j}+O(\rho^n)\Big)+n+O(1)\bigg]\bigg\}  &&\textrm{if }\wn<0, \label{non-zero winding formula 1}\\
	&\tilde{D}_n(\tilde{f})=D_n(f)\bigg\{1+z_n\bigg[-c(e^{\imath \theta_0}) \et{(n+1)}\sum_{j=1}^{\wn} \frac{\tilde{\Delta}_{\wn,n}(j)}{\Delta_{\wn,n}}\Big(\etm{j}+O(\rho^n)\Big)+n+O(1)\bigg]\bigg\} &&\textrm{if } \wn>0. \label{non-zero winding formula 2}
	\end{align}
\end{theorem}

Compared to the usual behavior of Toeplitz determinants without delta-function singularities, in this case we see the emergence of algebraic contributions in the matrix rank $n$. {The terms inside the sums in \eqref{non-zero winding formula 1} an \eqref{non-zero winding formula 2} are expected to be of the order of magnitude of the inverse of the coefficients $d_j^{(n)}$ for $|j|< |\wn|$.  Therefore, since $d_j^{(n)}=O(\rho^n)$ as $n\to\infty$, they are expected typically to grow faster than $\rho^{-n}$, with exponentially suppressed subleading corrections.}

{We are going to derive these theorems by relating the determinant to a linear problem, which in turn can be expressed as a linear functional equation, whose solution can be obtained in the limit of large $n$. The leading contributions to this solution, obtained using the Wiener-Hopf procedure, provide the asymptotic formulas above for the determinant. The estimates on the order of magnitude of the corrections to the terms we have calculated explicitly rely on the intuitively clear property that a small perturbation to the functional equation yields a change to the solution of a similar order of magnitude. For the case of a zero winding number of the symbol we have also provided a rigorous proof of this statement, which makes Theorem \ref{tm zero winding} rigorous. For the case of a non-zero winding number we have not provided a rigorous proof, but the corrections are plausible for the same reasons and we have confirmed Theorem \ref{tm non-zero winding} numerically on a few relevant examples. The rigorous analysis of corrections is better established in the framework of the Riemann-Hilbert problem, an alternative technique that can be used, which will be the subject of another work.}

{We derive the theorems in section \ref{section derivation}.} To give a concrete example of applications of these theorems (and to explicitly show the unusual behavior of these determinants) in section \ref{section application} we compute the spin-correlation functions in the lowest energy band of the frustrated quantum XY chain in zero magnetic field, and the ground state magnetization.

\section{Derivation of the Theorems}\label{section derivation}

\subsection{Linear Problem}
The first step in the derivation of the theorems is to use \eqref{elements tilde} and the basic property that determinant is alternating mutilinear function of its columns, to expand $\tilde{D}_n(\tilde{f})$ as
\begin{equation}\label{expanded}
\tilde{D}_n(\tilde{f})=D_n(f)+z_nf(e^{\imath \theta_0})\sum_{j=1}^{n} e^{\imath (j-1)\theta_0} D_{n,j}(f),
\end{equation} 
where by $D_{n,j}(f)$ we denote the determinant obtained by replacing the column $j$ in $D_n(f)$ by the column vector {$(1,e^{-\imath\theta_0},e^{-\imath 2\theta_0},...,e^{-\imath(n-1)\theta_0})^{\textrm{T}}$ }.

We assume that there is $n_0\in\mathbb{N}$ such that $D_n(f)\neq0$ for $n\geq n_0$. In the case $\wn=0$ this assumption is justified by the Szeg\H{o} theorem, while for nonzero $\wn$ we restrict ourselves to symbols for which this assumption holds. This assumption implies that there exists a unique solution $x_j^{(n)}$, $j=0,1,...,n-1$, of the linear problem
\begin{equation}\label{linear problem}
\sum_{k=0}^{n-1} f_{j-k}x_k^{(n)}=e^{-\imath\theta_0j}, \quad \textrm{for } j=0,1,...,n-1,
\end{equation}
that is by Cramer's rule given by
\begin{equation}
x_j^{(n)}=\frac{D_{n,j+1}(f)}{D_n(f)} \; .
\end{equation}
This solution can be inserted in \eqref{expanded} to get
\begin{equation}
\tilde{D}_n(\tilde{f})=D_n(f)\Big[{1}+z_nf(e^{\imath \theta_0})\sum_{j=0}^{n-1} e^{\imath \theta_0j} x_j^{(n)} \Big].
\end{equation}
Defining the analytic function
\begin{equation}\label{X definition}
X^{(n)}(z)=\sum_{j=0}^{n-1}x_j^{(n)}z^j
\end{equation}
we have
\begin{equation}\label{determinant in terms of X}
\tilde{D}_n(\tilde{f})=D_n(f)\Big(1+z_nf(e^{\imath \theta_0})X^{(n)}(e^{\imath\theta_0}) \Big) .
\end{equation}

We are thus going to find an asymptotic formula for $X^{(n)}(z)$, and hence for $\tilde{D}_n(\tilde{f})$, by using the Wiener-Hopf procedure, similar to the one of \cite{Wu1966,TheTwo-DimensionalIsingModel} used to compute the spin correlation functions of the Ising model. The prerequisites and details are given in the following sections.

{Let us comment that the presented method of relating the determinant to a linear problem would not work with more than one delta function in the symbol \eqref{symbol}. In arriving to expression \eqref{expanded} we have used the property that different column vectors resulting from the delta function in the symbol are one scalar multiple of the other. This has lead to many cancellations, since determinant is an alternating function. With more delta functions, column vectors arising from different delta functions could not anymore be related simply by scalar multiplication, which prevents cancellations.
}

\subsection{Equivalent Problem}\label{subsec equivalent}
For a function $g(z)=\sum_{j=-\infty}^{\infty}g_j z^j$, defined and analytic in an annulus $\rho_{-}<|z|<\rho_+$ including the unit circle, we define its components
\begin{equation}\label{components definition}
\left[g\right]_-(z)=\sum_{j=1}^{\infty}g_{-j}z^{-j} ,\quad \left[g\right]_+(z)=\sum_{j=0}^{\infty}g_jz^j .
\end{equation}
The function $\left[g\right]_-$ is analytic on {$|z|>\rho_-$}, while the function $\left[g\right]_+$ is analytic on {$|z|<\rho_+$}. For definiteness, in the following we are going to restrict the domain of these functions to the annulus $\rho_-<|z|<\rho_+$, where both are analytic. 

As shown in Appendix \ref{appendix existence}, defining the coefficients
\begin{equation}\label{y definition}
y_j^{(n)}=\begin{cases}
e^{-\imath j\theta_0} , & j\in\{0,1,...,n-1\}\\
0, & j\in\mathbb{Z}\backslash\{0,1,...,n-1\}
\end{cases}
\end{equation}
and the analytic function
\begin{equation}\label{Y definition}
Y^{(n)}(z)=\sum_{j=0}^{n-1}y_j^{(n)}z^j 
\end{equation}
solving the linear problem \eqref{linear problem} is equivalent to finding functions $X^{(n)}, U^{(n)}, V^{(n)}$, defined and analytic on an annulus including the unit circle, satisfying
\begin{equation}\label{equivalent problem}
fX^{(n)}=Y^{(n)}+U^{(n)}z^{n}+V^{(n)}
\end{equation}
and having the properties
\begin{equation}\label{properties}
\big[X^{(n)}\big]_-=0,\quad \big[X^{(n)}z^{-n}\big]_+=0,\quad \big[U^{(n)}\big]_-=0, \quad \big[V^{(n)}\big]_+=0 \; .
\end{equation}
With $D_n(f)\neq 0$ the solution exists and is unique, with $X^{(n)}$ corresponding to \eqref{X definition}. By solving this equivalent problem we find the asymptotic formula for $X^{(n)}(\et{})$ in \eqref{determinant in terms of X}.

\subsection{Evaluating the components}\label{subsec evaluating}

Following the standard Wiener-Hopf approach, we seek the solution of \eqref{equivalent problem} by exploiting the different analytical properties of the different components of the functions appearing in it.
The components \eqref{components definition} can be evaluated as the following integrals. Let $z$ belong to the annulus  $\rho_-<|z|<\rho_+$, and let $\rho_1\in(\rho_-,|z|), \rho_2\in(|z|,\rho_+)$. Then
\begin{equation}\label{plus minus integral representation}
\left[g\right]_-(z)=\frac{1}{2\pi\imath}\oint_{|w|=\rho_1}\frac{g(w)}{z-w} dw , \qquad \left[g\right]_+(z)=\frac{1}{2\pi\imath}\oint_{|w|=\rho_2}\frac{g(w)}{w-z} dw .
\end{equation}
These formulas can be shown by summing, in accordance with definition \eqref{components definition}, the Laurent series coefficients
\begin{equation}\label{Laurent coefficients}
g_k=\frac{1}{2\pi\imath}\oint_{|w|=\rho_1}g(w)\frac{dw}{w^{k+1}}=\frac{1}{2\pi\imath}\oint_{|w|=\rho_2}g(w)\frac{dw}{w^{k+1}}
\end{equation}
and interchanging the sum and the integral, which is valid since the Laurent series is uniformly convergent on every closed subannulus in the interior of its annulus.

In the derivation of the theorems we are going to encounter functions $G$, analytic on the annulus $\rho_-<|z|<\rho_+$, that are of the form
\begin{equation}\label{definition G}
G(z)=\frac{g(z)-g(e^{\imath\theta_0})}{z-e^{\imath\theta_0}} \quad \textrm{for } z\neq e^{\imath\theta_0}, \qquad G(e^{\imath\theta_0})=\diff{g(z)}{z}\bigg|_{z=e^{\imath\theta_0}},
\end{equation}
where $g$ is analytic on the same annulus. For instance, the function \eqref{Y definition} satisfies
\begin{equation}\label{Y explicitly}
Y^{(n)}(z)=\etm{(n-1)}\frac{z^n-\et{n}}{z-\et{}} , \quad Y^{(n)}(\et{})=n=\etm{(n-1)} \diff{{z^n}}{z}\bigg|_{z=\et{}}.
\end{equation}

For $z\neq\et{}$, it will be convenient to consider the function $G$ as a sum of two functions with a singularity on the unit circle, or, as we are about to show, the sum of two functions analytical inside/outside the unit circle. We thus need to introduce another structure. Let $\mathcal{G}$ be defined by the rule
\begin{equation}
\mathcal{G}(z)=\frac{g(z)}{z-\et{}}  \qquad \textrm{for } \rho_-<|z|<\rho_+,\ z\neq \et{},
\end{equation}
with $g$ being a function analytic on the annulus $\rho_-<|z|<\rho_+$. The function $\mathcal{G}$ is thus analytic on the annuli $\rho_-<|z|<1$ and $1<|z|<\rho_+$. Its Laurent series coefficients are different on two different annuli, let us denote them by
\begin{equation}
\mathcal{G}(z)=\sum_{j=-\infty}^{\infty}\mathcal{G}^{\ts}_j z^j \quad \textrm{for } \rho_-<|z|<1, \qquad \mathcal{G}(z)=\sum_{j=-\infty}^{\infty}\mathcal{G}^{\tl}_j z^j \quad \textrm{for } 1<|z|<\rho_+ .
\end{equation}
There are thus two ways of defining $+$ and $-$ components of the functions $\mathcal{G}$. We define
\begin{align}
\big[\mathcal{G}\big]_-^{\ts}(z)=\sum_{j=1}^{\infty}\mathcal{G}^{\ts}_{-j} z^{-j} , \quad \big[\mathcal{G}\big]_+^{\ts}(z)=\sum_{j=0}^{\infty} \mathcal{G}^{\ts}_{j}z^j, \\
\big[\mathcal{G}\big]_-^{\tl}(z)=\sum_{j=1}^{\infty}\mathcal{G}^{\tl}_{-j} z^{-j} , \quad \big[\mathcal{G}\big]_+^{\tl}(z)=\sum_{j=0}^{\infty} \mathcal{G}^{\tl}_{j}z^j.
\end{align}
In this work we are going to make use of the functions $\big[\mathcal{G}\big]_-^{\ts}$ and $\big[\mathcal{G}\big]_+^{\tl}$, both of which are analytic on $\rho_-<|z|<\rho_+$. We are going to use the obvious notation
\begin{equation}
\Big[\frac{g}{z-\et{}}\Big]_-^{\ts}=\big[\mathcal{G}\big]_-^{\ts}, \quad \Big[\frac{g}{z-\et{}}\Big]_+^{\tl}=\big[\mathcal{G}\big]_+^{\tl}.
\end{equation}

Analogously to \eqref{plus minus integral representation} we have the integral representation
\begin{equation}\label{integral representatino ts tl}
\begin{split}
&\Big[\frac{g}{z-\et{}}\Big]_-^{\ts}(z)=\frac{1}{2\pi\imath}\oint_{|w|=\rho_1}\frac{g(w)}{(w-\et{})(z-w)} dw , \quad \textrm{where } \rho_1\in\big(\rho_{-}, \min\{1,|z|\}\big),\\
&\Big[\frac{g}{z-\et{}}\Big]_+^{\tl}(z)=\frac{1}{2\pi\imath}\oint_{|w|=\rho_2}\frac{g(w)}{(w-\et{})(w-z)} dw , \quad \textrm{where } \rho_2\in\big( \max\{1,|z|\},\rho_+ \big).
\end{split}
\end{equation}
Expanding $1/(w-\et{})$ under integrals into series and interchanging the series and the integral, we get the following representation:
\begin{equation}\label{plus minus series}
\begin{split}
&\bigg[\frac{g}{z-\et{}}\bigg]_-^{\ts}(z)=\sum_{j=0}^{\infty} \etm{(j+1)}\big[gz^j\big]_-(z) , \\
& \bigg[\frac{g}{z-\et{}}\bigg]_+^{\tl}(z)=\sum_{j=0}^{\infty} \et{j}\big[gz^{-j-1}\big]_+(z).
\end{split}
\end{equation}

With the introduced definitions, for the function $G$ defined by \eqref{definition G} we have
\begin{equation}\label{components G}
\begin{split}
&\big[G\big]_-=\Big[\frac{g-g(\et{})}{z-\et{}}\Big]_-^{\ts}=\Big[\frac{g}{z-\et{}}\Big]_-^{\ts},\\
&\big[G\big]_+=\Big[\frac{g-g(\et{})}{z-\et{}}\Big]_+^{\tl}=\Big[\frac{g}{z-\et{}}\Big]_+^{\tl},
\end{split}
\end{equation}
where the first equality follows immediately from the integral representations and the second is obtained using
\begin{equation}
\Big[\frac{g(\et{})}{z-\et{}}\Big]_-^{\ts}=0 , \quad \Big[\frac{g(\et{})}{z-\et{}}\Big]_+^{\tl}=0 ,
\end{equation}
which follows from \eqref{plus minus series}. From \eqref{components G} it follows
\begin{equation}\label{decomposition G}
G=\Big[\frac{g}{z-\et{}}\Big]_-^{\ts}+\Big[\frac{g}{z-\et{}}\Big]_+^{\tl}.
\end{equation}
We clearly have also the following linear property. Let $h$ be a function analytic on the same annulus as $g$. Then
\begin{equation}
\big[hG\big]_-=\Big[\frac{hg}{z-\et{}}\Big]_-^{\ts}-\Big[\frac{hg(\et{})}{z-\et{}}\Big]_-^{\ts}, \qquad
\big[hG\big]_+=\Big[\frac{hg}{z-\et{}}\Big]_+^{\tl}-\Big[\frac{hg(\et{})}{z-\et{}}\Big]_+^{\tl}.
\end{equation}

The definitions we have introduced are going to be useful because of the following two elementary lemmas that we state and prove.

\begin{lemma}\label{lemma components t}
	Let $g$ be a function analytic on an annulus $\rho_-<|z|<\rho_+$, that includes the unit circle. Then for $z\neq\et{}$ belonging to the annulus
		\begin{equation}
	\Big[\frac{g}{z-\et{}}\Big]_-^{\ts}(z)=\frac{[g]_-(z)-[g]_-(\et{})}{z-\et{}}, \qquad \Big[\frac{g}{z-\et{}}\Big]_+^{\tl}(z)=\frac{[g]_+(z)-[g]_+(\et{})}{z-\et{}}.
	\end{equation}
\end{lemma}
\begin{proof}
	Let us prove the first equality. From the representation \eqref{plus minus series} it follows
	\begin{equation}
	\left[\frac{[g]_+}{z-\et{}}\right]_-^{\ts}=0,
	\end{equation}
	which implies
	\begin{equation}\label{step lemma G 1}
	\left[\frac{g}{z-\et{}}\right]_-^{\ts}=\left[\frac{[g]_-}{z-\et{}}\right]_-^{\ts}.
	\end{equation}

We now define a function $G^{(-)}$ as 
\begin{equation}
G^{(-)}(z)=\frac{[g]_-(z)-[g]_-(e^{\imath\theta_0})}{z-e^{\imath\theta_0}} \quad \textrm{for }  z\neq e^{\imath\theta_0}, \ \rho_-<|z|<\rho_+,  \qquad G^{(-)} (e^{\imath\theta_0})=\diff{}{z}[g]_-(z)\bigg|_{z=e^{\imath\theta_0}},
\end{equation}
and, using the decomposition \eqref{decomposition G} for this function, we have
\begin{equation}\label{step lemma G}
  G^{(-)} = \left[\frac{[g]_-}{z-\et{}}\right]_-^{\ts}
  +\left[\frac{[g]_-}{z-\et{}}\right]_+^{\tl}
  =\left[\frac{[g]_-}{z-\et{}}\right]_-^{\ts} ,
\end{equation}
where the last equality follows from \eqref{plus minus series}. Combining \eqref{step lemma G 1} and \eqref{step lemma G} proves the first equality of the lemma. The second equality is proven in an analogous way.
\end{proof}

\begin{lemma}\label{lemma plus minus components}
	Let $(g_n)_{n\in\mathbb{N}}$ be a sequence of functions analytic on an annulus $\rho_-<|z|<\rho_+$, that includes the unit circle, and let $\rho_-<\rho_1<1<\rho_2<\rho_+$. Moreover, let $(s_n)_{n\in\mathbb{N}}$ be a sequence of positive numbers.
	\begin{enumerate}[label=(\alph*)]
		\item\label{part a} If $g_n=O(s_n)$ uniformly in $z$ at the circle $|z|=\rho_1'$, for some $\rho_1'\in(\rho_-,\rho_1)$, then
		\begin{equation}
		\big[g_n\big]_-(z)=O(s_n),\quad  \bigg[\frac{g_n}{z-\et{}}\bigg]_-^{\ts}(z)=O(s_n)
		\end{equation}
		uniformly in $z$ on $\rho_1\leq|z|\leq \rho_2$.
		 \item If $g_n=O(s_n)$ uniformly in $z$ at the circle $|z|=\rho_2'$, for some $\rho_2'\in(\rho_2,\rho_+)$, then
		\begin{equation}
		\big[g_n\big]_+(z)=O(s_n),\quad \bigg[\frac{g_n}{z-\et{}}\bigg]_+^{\tl}(z)=O(s_n)
		\end{equation}
		uniformly in $z$ on $\rho_1\leq|z|\leq \rho_2$.
	\end{enumerate}
\end{lemma}

\begin{proof}
	Let us prove the first part of \ref{part a}. By assumption there is $K>0$ such that
	\begin{equation}\label{step proof lemma}
	|g_n(z)|\leq K s_n \quad \textrm{for all $z$ on the circle $|z|=\rho_1'$}. 
	\end{equation}
	For $\rho_1\leq|z|\leq\rho_2$ we have, from \eqref{plus minus integral representation}, the integral representation
	\begin{equation}
	\big[g\big]_-(z) =\frac{1}{2\pi\imath} \oint_{|w|=\rho_1'}\frac{g(w)}{z-w}dw.
	\end{equation}
	Then from the assumption \eqref{step proof lemma} and using $|z-w|\geq \rho_1-\rho_1'$ it follows
\begin{equation}
\Big|\big[g\big]_-(z)\Big|\leq \frac{K \rho_1'}{\rho_1-\rho_1'} s_n,
\end{equation}
which means that $\big[g\big]_-(z) = O(s_n)$ uniformly in $z$ on $\rho_1\leq|z|\leq \rho_2$. The other parts of the lemma are proven in an analogous way by using the integral representations \eqref{plus minus integral representation} and \eqref{integral representatino ts tl}. 
\end{proof}

\subsection{Solution for the zero winding number case}\label{sec zero winding}

Having introduced the tools for the separation in components of the various functions, we can present the solutions. In the case $\wn=0$, on the basis of a Wiener-Hopf procedure, presented in Appendix \ref{appendix solution}, we construct the following functions, which are the specialization of \eqref{obtained x 1}, \eqref{obtained u 1}, and \eqref{obtained V} and are defined and analytic on the annulus \eqref{rho plus minus}, given for $z\neq\et{}$ by
\begin{align}\label{functions 1}
& X_1^{(n)}(z)=\etm{(n-1)}a_+^{-1}(\et{})a_+^{-1}(z) \frac{z^n c(z)-\et{n}c(\et{})}{z-\et{}} \; , \\
& U_1^{(n)}(z)=-\etm{(n-1)}a_+(z)\frac{a_+^{-1}(z)-a_+^{-1}(\et{})}{z-\et{}} \; ,  \\
&V_1^{(n)}(z)=\et{}a_-(z)\frac{a_-^{-1}(z)-a_-^{-1}(\et{})}{z-\et{}} \; , 
\end{align}
and for $z=\et{}$ by continuity. It's easy to see that these functions satisfy the equation \eqref{equivalent problem}
\begin{equation}
aX_1^{(n)}=Y^{(n)}+U_1^{(n)}z^{n}+V_1^{(n)} \; ,
\end{equation}
where, in this case, according to \eqref{symbol f}, $f=a$, and the function $a_\pm$ in \eqref{functions 1} have been defined in \eqref{apmdef} and we remind that $c=\app \amm$.

Let $\rho$ be defined by \eqref{rho definition}. A straightforward application of Lemma \ref{lemma components t} and Lemma \ref{lemma plus minus components} yields the properties
\begin{equation}\label{properties 1}
\big[X_1^{(n)}\big]_-=O(\rho^n), \quad \big[X_1^{(n)}z^{-n}\big]_+=O(\rho^n), \quad
\big[U_1^{(n)}\big]_-=0 , \quad \big[V_1^{(n)}\big]_+=0 ,
\end{equation}
where $O(\rho^n)$ holds on $\rho\leq|z|\leq\rho^{-1}$, uniformly in $z$. For example, let us show the first property. We have
\begin{equation}
\big[X^{(n)}_1\big]_-=\etm{(n-1)}\apm(\et{})\bigg[\frac{z^na_-^{-1}}{z-\et{}}\bigg]_-^{\ts}-\et{}a_-^{-1}(\et{})\bigg[\frac{a_+^{-1}}{z-\et{}}\bigg]_-^{\ts}.
\end{equation}
Applying Lemma \ref{lemma plus minus components} on the first term, with $\rho_1'\in(\rho_{-},\rho),\ \rho_1=\rho,\ \rho_2=\rho^{-1}, s_n=\rho^n$, we see that it is equal to $O(\rho^n)$ on $\rho\leq|z|\leq \rho^{-1}$, uniformly in $z$. The second term is zero, by Lemma \ref{lemma components t}. The other properties in \eqref{properties 1} are shown in an analogous way.

The properties \eqref{properties 1} should be compared with \eqref{properties}. {They imply that the function $X_1^{(n)}$ with the components $j<0$ and $j\geq n$ removed, i.e. the function $X_2^{(n)}=X_1^{(n)}-[X_1^{(n)}]_--z^n[X_1^{(n)}z^{-n}]_+$, is a solution of the problem
\begin{align}
& aX_2^{(n)}=Y_2^{(n)}+U_2^{(n)}z^{n}+V_2^{(n)}\\
& \big[X_2^{(n)}\big]_-=0,\quad \big[X_2^{(n)}z^{-n}\big]_+=0,\quad \big[U_2^{(n)}\big]_-=0, \quad \big[V_2^{(n)}\big]_+=0 \;,
\end{align}
where {$[\yw_2]_-=0, \ [\yw_2 z^{-n}]_+=0 $}, and
\begin{equation}
\yw_2(z)-\yw(z)=O(\rho^n) \quad \textrm{uniformly in $z$ on the unit circle $|z|=1$.}
\end{equation}
Thus, we have constructed a linear problem whose corresponding functional equation is a $O(\rho^n)$ perturbation of the original one on the unit circle.} {It is intuitively clear that the solution $\xw$ will result into a $O(\rho^n)$ perturbation of $\xw_2$ on the unit circle.} {This passage is made rigorous by applying Lemma \ref{lemma rigor}, stated and proven in Appendix \ref{Appendix rigor}, on $\xw$ and $\xw_2$. It follows
\begin{equation}
\xw_2(z)-\xw(z)=O(\rho^n) \quad \textrm{uniformly in $z$ on the unit circle $|z|=1$}
\end{equation}
and therefore, since $\xw_2(z)-\xw_1(z)=O(\rho^n)$ uniformly in $z$ on the unit circle, we have
\begin{equation}
X^{(n)}(z)=X_1^{(n)}(z)+O(\rho^n) \quad \textrm{uniformly in $z$ on the unit circle $|z|=1$.}
\end{equation}
}

This, together with \eqref{functions 1}, gives
\begin{equation}
X^{(n)}(\et{})=a_+^{-1}(\et{})a_-^{-1}(\et{}) n+\et{}a_+^{-2}(\et{}) \diff{c(z)}{z}\bigg|_{z=\et{}}+O(\rho^n).
\end{equation}
It follows
\begin{equation}\label{step a x zero}
a(\et{})X^{(n)}(\et{})=n-\imath\diff{}{\theta}\log c(e^{\imath\theta})\bigg|_{\theta=\theta_0}+O(\rho^n).
\end{equation}
Theorem \ref{tm zero winding} follows from \eqref{determinant in terms of X} and \eqref{step a x zero}. Let us comment here that the leading term in this solution was already determined in \cite{Dong_2016,DongPRE}, {but without rigor}, with a cavalier use of the component analysis and an improper analytical continuation. Most of all, the approach employed there does not allow to treat the non-zero winding number case.

\subsection{Solution for the non-zero winding number case}\label{sec non zero winding}

Let us assume that $\wn>0$. The result for $\wn<0$ follows from this one by transposing the original Toeplitz {matrix in} \eqref{first} and doing the integral transformation $\theta\to-\theta$. On the basis of the Wiener-Hopf procedure presented in Appendix \ref{appendix solution}, as a specialization of \eqref{obtained x 1}, \eqref{obtained u 1} and \eqref{obtained V} we construct the functions $X^{(n)}_1, U^{(n)}_1, V^{(n)}_1$, analytic on the annulus \eqref{rho plus minus}. For $z\neq\et{}$, $\rho_-<|z|<\rho_+$, they are defined by the rule 
\begin{align}
\begin{split}\label{X 1 wn}
&\xw_1(z)z^\wn=\etm{(n+\wn-1)}\apm(z)\apm(\et{}) \frac{c(z)z^{n+\wn}-c(\et{})\et{(n+\wn)}}{z-\et{}}\\
&\qquad+\etm{(n-1)}\amm(z)z^{n+\wn}\sum_{k=0}^{\wn-1}(\apm)_k\frac{z^{k-\wn}-\et{(k-\wn)}}{z-\et{}}+\apm(z)\sum_{k=1}^{\wn}\ucoeff{k}\big[c \: z^{n+\wn-k}\big]_+(z),
\end{split}\\
& \uw_1 z^{-\wn}=-\etm{(n-1)}\app(z)\frac{[\apm z^{-\wn}]_+(z)-[\apm z^{-\wn}]_+(\et{})}{z-\et{}} +\app(z)\sum_{k=1}^{\wn}\ucoeff{k} z^{-k},\\
& \vw_1(z)=\et{}\amp(z)\frac{\amm(z)-\amm(\et{})}{z-\et{}}-\amp(z)\sum_{k=1}^{\wn}\ucoeff{k}\big[c \: z^{n+\wn-k}\big]_-(z),
\end{align}
and for $z=\et{}$ by continuity. Here $\coeff{1},\coeff{2},...,\coeff{\wn}\in\mathbb{C}$ are for the moment unspecified, and it is simple to check that for any choice of them the functions above satisfy \eqref{equivalent problem} as
\begin{equation}
a z^\wn X_1^{(n)}=Y^{(n)}+U_1^{(n)}z^{n}+V_1^{(n)}  .
\end{equation}

Let $\rho$ be defined by \eqref{rho definition}. As in the previous section, a straightforward application of Lemma \ref{lemma components t} and Lemma \ref{lemma plus minus components} yields
\begin{align}\label{step x nu}
\big[X_1^{(n)}z^\wn\big]_-(z)=O(\rho^n) , \quad \big[U_1^{(n)}\big]_-=0 , \quad \big[V_1^{(n)}\big]_+=0,
\end{align}
where $O(\rho^n)$ holds on $\rho\leq|z|\leq\rho^{-1}$, uniformly in $z$. The coefficients $\coeff{1},\coeff{2},...,\coeff{\wn}$ are chosen to satisfy
\begin{equation}\label{components zero}
\big(X^{(n)}_1 z^\wn\big)_j=O(\rho^n)  \quad \textrm{for } j=0,1,...,\wn-1,
\end{equation}
thus extending \eqref{step x nu} to also $\big[X_1^{(n)}\big]_-=O(\rho^n).$ 

Thus, one computes the components $\big(X^{(n)}_1 z^\wn\big)_j$ using \eqref{Laurent coefficients}, by integrating at the circle $|w|=\rho$, and imposes \eqref{components zero}. As shown in Appendix \ref{appendix solution}, this procedure results in
\begin{equation}\label{coefficients}
\coeff{k}=-a_-^{-1}(\et{})\etm{(\wn-1)}\frac{\tilde{\Delta}_{\wn,n}(k)}{\Delta_{\wn,n}} ,
\end{equation}
where $\Delta_{\wn,n}$ is defined by \eqref{Delta} and $\tilde{\Delta}_{\wn,n}(k)$ by \eqref{Delta tilde l}.

We have shown so far
\begin{equation}\label{properties 1 wn}
\big[\xw_1\big]_-=O(\rho^n), \quad
\big[\uw_1\big]_-=0 , \quad \big[\vw_1\big]_+=0 ,
\end{equation}
which should be compared to \eqref{properties}. It remains to discuss $\big[X_1^{(n)}z^{-n}\big]_+$. Application of Lemmas \ref{lemma components t} and \ref{lemma plus minus components} yields
\begin{equation}
\big[\xw_1z^{-n}\big]_+=O(\rho^n)+\sum_{k=1}^{\wn} \coeff{k}\bigg[a_+^{-1}z^{-(n+\wn)}\big[cz^{n+\wn-k}\big]_+\bigg]_+ .
\end{equation}
We have further
\begin{equation}
\coeff{k}\bigg[\apm z^{-(n+\wn)} \big[cz^{n+\wn-k}\big]_+\bigg]_+=-\coeff{k}\bigg[\apm z^{-(n+\wn)}\big[cz^{n+\wn-k}\big]_-\bigg]_+=-\coeff{k} O(\rho_1^n)
\end{equation}
where in the last equality we have applied Lemma \ref{lemma plus minus components} two successive times for for some $\rho_1\in(\rho_-,\rho)$, and $O(\rho_1^{2n})$ holds on $\rho\leq|z|\leq \rho^{-1}$, uniformly in $z$. Now, assuming that the condition \eqref{condition tm} of Theorem \ref{tm non-zero winding} holds, using \eqref{coefficients} we get
\begin{equation}
\coeff{k}O(\rho_1^{2n})=O\big((\rho_1/\rho)^2\big)
\end{equation}
Defining
\begin{equation}
\sigma=\max\{(\rho_1/\rho)^2, \rho\}
\end{equation}
we have thus
\begin{equation}\label{properties 1 wn last}
\big[\xw_1\big]_-(z)=O(\sigma^n), \quad \big[\xw_1 z^{-n}\big]_+=O(\sigma^n), \quad
\big[\uw_1\big]_-=0 , \quad \big[\vw_1\big]_+=0 ,
\end{equation}
where $O(\sigma^n)$ holds on $\rho\leq |z|\leq \rho^{-1}$, uniformly in $z$.  These properties should be compared with \eqref{properties}. {Because of the same reasoning as for the zero winding number case we expect}
\begin{equation}
X^{(n)}(z)=X_1(z)+O(\sigma^n)
\end{equation}
on the unit circle $|z|=1$. {In this case we cannot provide a rigorous proof for this estimate, but we find it very natural and we have checked numerically that the final result, i.e. Theorem \ref{tm non-zero winding}, holds in several relevant cases.}

It follows
\begin{equation}
\begin{split}\label{step X  et wn}
X^{(n)}(\et{})&\et{\wn}=\etm{(n+\wn-1)}a_+^{-2}(\et{}) \diff{}{z}\Big(c(z)z^{n+\wn}\Big)\bigg|_{z=\et{}}\\
&+\et{(\wn+1)}a_-^{-1}(\et{})\sum_{k=0}^{\wn-1}\big(a_+^{-1}\big)_k\diff{}{z}z^{k-\wn}\bigg|_{z=\et{}}+a_+^{-1}(\et{})\sum_{k=1}^{\wn}\coeff{k} \big[c z^{n+\wn-k}\big]_+(\et{})+O(\sigma^n)
\end{split}
\end{equation}
from which we get
\begin{equation}
\begin{split}
\xw(\et{})\et{\wn}  &= 
\apm(\et{})\amm(\et{}) (n+\wn)
+\et{}a_+^{-2}(\et{}) \diff{c(z)}{z}\bigg|_{z=\et{}} \\
&+a_-^{-1}(\et{})\sum_{k=0}^{\wn-1} (k-\wn)\big(a_+^{-1}\big)_k \et{k}
+a_+^{-1}(\et{})\sum_{k=1}^{\wn}\coeff{k} \big[c z^{n+\wn-k}\big]_+(\et{})+O(\sigma^n).
\end{split}
\end{equation}
Lemma \ref{lemma plus minus components} gives a simplification
\begin{equation}\label{simplification}
\big[c z^{n+\wn-k}\big]_+(\et{})=c(\et{})\et{(n+\wn-k)}+O(\rho^n)
\end{equation}
from which it follows
\begin{equation}
\xw(\et{})\et{\wn} =\amm(\et{})\sum_{k=1}^{\wn} \coeff{k}\big(\et{(n+\wn-k)}+O(\rho^n)\big)+\apm(\et{})\amm(\et{})n+O(1).
\label{before tm non zero}
\end{equation}
Theorem \ref{tm non-zero winding} follows directly from \eqref{determinant in terms of X} and \eqref{before tm non zero}, where the result for the case $\wn <0$ descends from the one for $\wn>0$ by transposing the original Toeplitz {matrix in} \eqref{first} and making the integral transformation $\theta \to -\theta$.

\section{Application of the results: Frustrated quantum XY chain in zero field}\label{section application}

As an example of a concrete application of our results we compute the lowest energy band spin-correlation functions and the ground state magnetization for the frustrated quantum spin chain defined by the Hamiltonian
\begin{equation}\label{Hamiltonian}
	H=\sum_{j=1}^{N} \Big(\sigma_j^x\sigma_{j+1}^x -\lambda \sigma_j^y\sigma_{j+1}^y\Big),
\end{equation} 
where $\lambda\in(0,1)$ is the anisotropy parameter, $N=2M+1$ is the number of lattice sites, which is imposed to be odd, $\sigma^\alpha$ for $\alpha=x,y,z$, are Pauli matrices, and periodic boundary conditions are imposed ($\sigma_j^\alpha=\sigma_{j+N}^\alpha$). This kind of models, known as XY chains, have been introduced in \cite{LIEB1961407} and its ground state spin-correlation functions have been computed in \cite{McCoy1968,LIEB1961407,McCoy2}. However, in \eqref{Hamiltonian} we take the dominant interaction (along the $x$ component of the spins) to favor antiferromagnetic order. This choice, coupled with the periodic boundary conditions on a ring with an odd number of sites $N$, introduces \textit{boundary frustration} different properties, which have been studied in \cite{Dong_2016,DongPRE,Giampaolo:2018imc,mari2019frustration,mari2020frustration}.
The frustrated model is interesting also because it demonstrates that different boundary conditions and different parities in $N$ can result in different ground state order even in the thermodynamic limit ($N\to\infty$).

We briefly review some properties of the model, found in \cite{mari2019frustration,mari2020frustration}, to introduce the notation. Denoting by $\Pi^\alpha=\prod_{j=1}^{N}\sigma_j^\alpha$, for $\alpha=x,y,z$, the parity operators, we have that all three commute with the Hamiltonian \eqref{Hamiltonian} $\left[ \Pi^\alpha, H\right]=0$, but, with odd $N$, satisfy a non-commuting algebra $\left[ \Pi^\alpha, \Pi^\beta \right]=\imath\,\varepsilon^{\alpha,\beta,\gamma} 2 (-1)^\frac{N-1}{2} \Pi^\gamma$, which is essentially SU(2). More interestingly for us, the parities anticommute $\left\{ \Pi^\alpha, \Pi^\beta \right\} = \delta_{\alpha,\beta}$. Because of these symmetries, every eigenstate of the spin chain \eqref{Hamiltonian} is at least two-fold degenerate. In fact, if $\ket{\Omega}$ is an eigenstate of \eqref{Hamiltonian} with, for instance, positive $z$-parity $\Pi^z \ket{\Omega}= \ket{\Omega}$, then $\Pi^x \ket{\Omega}$ has the same energy eigenvalue with respect to $H$ and opposite $z$-parity.

In each sector of given $z$-parity, the XY chain can be mapped exactly, although non-locally, to a system of free fermions \cite{Franchini:2016cxs}. This mapping allows to represent every state in a Fock space: one defines a vacuum $\ket{0^\pm}$
which is annihilated by fermionic operators $a_q$, with $q \in \Gamma^\pm$ belonging to a different set (of integers or half-integers) depending on the parity ($\Gamma^\pm=\left\{\frac{\pi}{N} \left(2j+\frac{1+\Pi^z}{2} \right): j=0,1,...,N-1\right\}$), $a_q \ket{0^\pm}=0$  for all $q\in\Gamma^\pm$, and applies the Bogoliubov creation operators $a_q^\dagger$ to create all other states. Only states with a number of excitations compatible with the parity are admissible in the Hilbert space of \eqref{Hamiltonian}: assigning zero excitations to the vacua $\ket{0^\pm}$, each $a_q^\dagger$ adds one. Even excitation states belong to the positive $z$-parity sector, while odd excitation ones have negative $z$-parity.

It is convenient to work just in a single $z$-parity (we will take $\Pi^z=-1$) and generate all (the degenerate) states with opposite parity through the application of $\Pi^x$. Due to the frustrated boundary conditions, the system is gapless with the energy gap between the states closing as $1/N^2$. This means that the lowest energy state is part of a band, spanned by the (single excitation) states $\ket{q}=a_q^\dagger\ket{0^-}$, which have negative $z$-parity $\Pi^z=-1$, and the states $\Pi^x\ket{q}=\imath(-1)^{(N-1)/2}\Pi^y\ket{q}$, which have the opposite $z$-parity $\Pi^z=1$. The energy of the states $\ket{q}$ and $\Pi^x\ket{q}$ is equal and the index $q$ is the momentum of the excitation, that can be related to lattice translations \cite{mari2020frustration}. The ground state, in particular, has momentum $q=0$, and is two-fold degenerate \cite{mari2019frustration}.
A generic ground state is, therefore, a superposition
\begin{equation}\label{ground state}
	\ket{g}=\cos\theta\ket{g^-}+\sin\theta\ e^{\imath\psi}\ket{g^+},
\end{equation}
where $\ket{g^-}=\ket{q=0}$, $\ket{g^+}=\Pi^x\ket{g^-}$, $\theta\in[0,2\pi)$ and $\psi\in[0,2\pi)$.

We are interested in the spin correlation functions $\bra{q}\sigma_1^x\sigma_{1+n}^x\ket{q}$ and $\bra{q}\sigma_1^y\sigma_{1+n}^y\ket{q}$, and in the ground state magnetizations $\bra{g}\sigma_1^x\ket{g}$ and $\bra{g}\sigma_1^y\ket{g}$. Note that, as a consequence of the symmetries and of the exact degeneracies, it is possible to have a spontaneous finite magnetization even for finite $N$.

\subsubsection{Spin-correlation functions}
Using the Majorana fermions representation of the spin operators and performing the Wick contractions as in \cite{DongPRE}, the spin correlation functions can be expressed in terms of Toeplitz determinants
\begin{equation}\label{two point}
	\bra{q} \sigma_1^\alpha\sigma_{1+n}^\alpha\ket{q}=(-1)^n\Big[\Big(\tilde{D}_n(\tilde{f})+\textrm{c.c.}\Big)-D_n(f)\Big],
\end{equation}
where $\textrm{c.c.}$ stands for the complex conjugate and
\begin{align}
	&\tilde{f}_j^{(n)}=f_j-\frac{1}{N}f(e^{\imath q})e^{-\imath q j},\\
	&f_j=\frac{1}{N}\sum_{\theta\in\Gamma^-} f(e^{\imath \theta})e^{-\imath j \theta}\stackrel{N\to\infty}{\simeq}\frac{1}{2\pi}\int_{0}^{2\pi}f(e^{\imath\theta})e^{-\imath j \theta} d\theta,\\
	&f(z)=a(z)z^\wn, \quad a(z)=\sqrt{\frac{1-\lambda z^{-2}}{1-\lambda z^2}}.\\
\end{align}
The winding number is $\wn=0$ for $\alpha=x$ and $\wn=2$ for $\alpha=y$. Comparing with \eqref{elements tilde} we see that $z_n=-1/N$, thus a constant with respect to $n$, although, from physical considerations, we must have $n<N/2$. 

We see that $a$ is analytic on $\lambda^{1/2}<|z|<\lambda^{-1/2}$ and by inspection we find
\begin{equation}
	\app(z)=\amm(z^{-1})=(1-\lambda z^2)^{-1/2}, \quad c(z)=b^{-1}(z)=\app(z)\amm(z)=\big[(1-\lambda z^2)(1-\lambda z^{-2})\big]^{-1/2}.
\end{equation}
The determinant $D_n(f)$ has already been computed in \cite{McCoy1968} because it determines the ground state spin-correlation functions in absence of frustration. For $\wn=0$ it is given by
\begin{align}
	&D_n(f)\\
	&=(1-\lambda^2)^{1/2}\Big[1+4\Big(\frac{\lambda}{1-\lambda^2}\Big)^2\frac{\lambda^n}{\pi n^2}\big(1+O(n^{-1})\big)\Big] , \quad \textrm{for $n=2m$ as $m \to\infty$, }\\
	&=(1-\lambda^2)^{1/2}\Big[1+2\frac{1+\lambda^2}{\lambda}\Big(\frac{\lambda}{1-\lambda^2}\Big)^2\frac{\lambda^n}{\pi n^2}\big(1+O(n^{-1})\big)\Big] , \quad \textrm{for $n=2m+1$ as $m \to\infty$. }
\end{align}
Applying Theorem \ref{tm zero winding}, with the term
\begin{equation}
	\diff{}{\theta}\log b(e^{\imath\theta})\bigg|_{\theta=q}=\frac{2\lambda  \sin q }{1+\lambda^2-2\lambda\cos2q}\in\mathbb{R}
\end{equation}
not contributing in \eqref{two point}, we find
\begin{align}\label{xx}
	&\bra{q} \sigma_1^x\sigma_{1+n}^x\ket{q}\\
	&=(1-\lambda^2)^{1/2}\Big[1+4\Big(\frac{\lambda}{1-\lambda^2}\Big)^2\frac{\lambda^n}{\pi n^2}\Big(1+O(n^{-1})\Big)\Big]\Big[1-\frac{2n}{N}\Big(1+O(\lambda^{\frac{n}{2}(1+\varepsilon)})\Big)\Big] , \quad \textrm{for $n=2m$ as $m \to\infty$, }\\
	&=-(1-\lambda^2)^{1/2}\Big[1+2\frac{1+\lambda^2}{\lambda}\Big(\frac{\lambda}{1-\lambda^2}\Big)^2\frac{\lambda^n}{\pi n^2}\big(1+O(n^{-1})\big)\Big] \Big[1-\frac{2n}{N}\Big(1+O(\lambda^{\frac{n}{2}(1+\varepsilon)})\Big)\Big], \quad \textrm{for $n=2m+1$ as $m \to\infty$, }
\end{align}
where $\varepsilon>0$ is arbitrarily small.

To compute $\bra{q} \sigma_1^y\sigma_{1+n}^y\ket{q}$ using Theorem \ref{tm non-zero winding} we need to find
\begin{equation}
	c_{-n}=\frac{1}{2\pi\imath} \oint_{|w|=1}\frac{w^{n-1}}{[(1-\lambda z^{2})(1-\lambda z^{-2})]^{1/2}}dw.
\end{equation}
Integrals of this type have been computed in \cite{McCoy2,McCoy1968,Wu1966,TheTwo-DimensionalIsingModel}, for the purpose of computing the ground state spin-correlation functions, using the properties of the hypergeometric functions. This one is given by
\begin{align}
	& c_{-n}=\frac{2^{1/2}}{(1-\lambda^2)^{1/2}} \frac{\lambda^{\frac{n}{2}}}{\sqrt{\pi n}}\Big(1+O(n^{-1})\Big) \quad \textrm{for $n=2m$ as $m\to\infty$.}\\
	&  c_{-n}=0 \quad \textrm{for $n=2m+1$}.
\end{align}
Applying Theorem \ref{tm non-zero winding}, where the condition \eqref{condition tm} of the theorem is satisfied for $\rho$ close to $\sqrt{\lambda}$, and using the result \eqref{FH nonzero} for $D_n(f)$, we find
\begin{align}\label{yy}
	&\bra{q} \sigma_1^y\sigma_{1+n}^y\ket{q}\\
	&=\frac{2}{1-\lambda}\frac{\lambda^n}{\pi n}\Big(1+O(n^{-1})\Big)+2^{5/2}\frac{\cos n q}{(1+\lambda^2-2\lambda \cos 2 q)^{1/2}}\frac{\lambda^{\frac{n}{2}}}{N\sqrt{\pi n}}\Big(1+O(n^{-1})\Big) \quad \textrm{for $n=2m$ as $m\to\infty$.}\\
	&=\frac{2}{1-\lambda}\frac{\lambda^n}{\pi n}\Big(1+O(n^{-1})\Big)+2^{3/2}\frac{\lambda^{-\frac{1}{2}}\cos [(n+1)q]+\lambda^{\frac{1}{2}}\cos[(n-1)q]}{(1+\lambda^2-2\lambda \cos 2 q)^{1/2}}\frac{\lambda^{\frac{n}{2}}}{N\sqrt{\pi n}}\Big(1+O(n^{-1})\Big) \quad \textrm{for $n=2m+1$ as $m\to\infty$.}
\end{align}
In the ground state ($q=0$) terms proportional to $1/N$ in \eqref{xx} and \eqref{yy} are due to the delta-function singularity in the symbol and make the difference between the frustrated model and the model without frustration (that is, periodic boundary conditions with $N=2M$ or free boundary conditions). In this case, the difference is relevant only at distances $n$ comparable to the system size $N$. Without these terms, the $x$ correlation function would converge exponentially fast to a saturation value as the distance between sites is increased, while the $y$ corrrelation decays to zero, reflecting a spontaneous magnetization developing only in the $x$ direction. The dependence \eqref{xx} implies, instead, that the correlations between the most distant spins, separated by $n=(N-1)/2$, decay as $1/N$ as we increase the (odd) system size $N$. This kind of behavior in frustrated quantum chains was first found in \cite{CampostriniRings,Campostrini_2015} and later further discussed, and checked numerically, in \cite{Dong_2016,DongPRE,Giampaolo:2018imc,mari2019frustration}.

\subsubsection{The ground state magnetization}
As discussed in \cite{mari2019frustration,mari2020frustration}, the magnetization in the ground state is mesoscopic (that is, finite in finite systems) and ferromagnetic, i.e. $\bra{g}\sigma_1^x\ket{g}=\bra{g}\sigma_j^x\ket{g}$ for $j=2,3,...,N$, in a generic ground state $\ket{g}$ defined by \eqref{ground state}. It is given by
\begin{align}
	&\bra{g}\sigma_1^x\ket{g}=\cos\psi\sin 2\theta \bra{g^-}\sigma_1^x\Pi^x\ket{g^-}\\
	&\bra{g}\sigma_1^y\ket{g}=(-1)^{\frac{N-1}{2}}\sin\psi\sin 2 \theta \bra{g^-}\sigma_1^y\Pi^y\ket{g^-}
\end{align}
The absolute values of the quantities $\bra{g}\sigma_1^\alpha\Pi^\alpha\ket{g}$, for $\alpha=x,y$, are the maximal values of the magnetization on the ground state manifold, and it has been shown in \cite{mari2019frustration} that these quantities can be expressed as Toeplitz determinants, as
\begin{equation}
	\bra{g^-}\sigma_1^\alpha\Pi^\alpha\ket{g^-}\stackrel{N\to\infty}{\simeq}(-1)^n \tilde{D}_n(\tilde{f}),
\end{equation}
where
\begin{align}
	&n=\frac{N-1}{2}, \quad\tilde{f}_j^{(n)}=f_j-\frac{2}{N},\\
	&f_j=\frac{1}{2\pi}\int_{0}^{2\pi}f(e^{\imath\theta})e^{-\imath j \theta} d\theta,\\
	&f(z)=a(z)z^\wn, \quad a(z)=\sqrt{\frac{1-\lambda z^{-1}}{1-\lambda z}}.\\
\end{align}
Here the winding number is $\wn=0$ for $\alpha=x$ and $\wn=1$ for $\alpha=y$. Theorems \ref{tm zero winding} and \ref{tm non-zero winding} can be applied with $z_n=1/N=1/(2n+1)$.

We proceed in a similar way to the previous section. The function $a(z)$ is analytic on $\lambda<|z|<\lambda^{-1}$, and by inspection we find
\begin{equation}
	\app(z)=\amm(z^{-1})=(1-\lambda z)^{-1/2}, \quad c(z)=b^{-1}(z)=\app(z)\amm(z)=\big[(1-\lambda z)(1-\lambda z^{-1})\big]^{-1/2}.
\end{equation}
The coefficients $c_{-n}$ are given by
\begin{equation}
	c_{-n}=\frac{1}{2\pi\imath} \oint_{|w|=1}\frac{w^{n-1}}{[(1-\lambda z)(1-\lambda z^{-1})]^{1/2}}dw= \frac{\lambda^n}{\sqrt{\pi n}}\Big(1+O(n^{-1})\Big).
\end{equation}
Applying Theorems \ref{tm zero winding} and \ref{tm non-zero winding} we find
\begin{align}
	&\bra{g^-}\sigma_1^x\Pi^x\ket{g^-}=(-1)^{\frac{N-1}{2}}\frac{1}{N}(1-\lambda^2)^{\frac{1}{4}}\Big(1+O(\lambda^{\frac{N}{2}(1+\varepsilon)})\Big),
	\label{xmag}\\
	&\bra{g^-}\sigma_1^y\Pi^y\ket{g^-}=\frac{2}{N}\frac{(1+\lambda)^{\frac{1}{4}}}{(1-\lambda)^{\frac{3}{4}}}\Big(1+O(\lambda^{\frac{N}{2}(1+\varepsilon)})\Big),
	\label{ymag}
\end{align}
where $N=2M+1$, as $M\to\infty$. Here $\varepsilon>0$ is arbitrarily small.

We remark that without frustration (that is, without the delta-function in the symbol) the Toeplitz determinant in \eqref{xmag} would approach a constant exponentially fast, while the one in \eqref{ymag} would similarly decay to zero, while with the delta-function they both show an algebraic decay in the matrix rank.
We conclude that both magnetizations go to zero as $N=2M+1,M\to\infty$, which is in a striking difference from the behavior of the model in the limit $N=2M, M\to\infty,$ and from the behavior of the model with free boundary conditions.

\section*{Acknowledgments}
We thank Alexander Its for reading the manuscript and providing his comments on it, Salvatore Marco Giampaolo for useful discussions and preliminary numerical analysis, and Alexander Abanov for his help. We acknowledge support from the European Regional Development Fund -- the Competitiveness and Cohesion Operational Programme (KK.01.1.1.06 -- RBI TWIN SIN) and from the Croatian Science Foundation (HrZZ) Projects No. IP--2016--6--3347.
FF also acknowledge support from the QuantiXLie Center of Excellence, a project co--financed by the Croatian Government and European Union through the European Regional Development Fund -- the Competitiveness and Cohesion (Grant KK.01.1.1.01.0004) and from the Croatian Science Foundation (HrZZ) Projects No. IP--2019--4--3321.

\appendix

\begin{appendices}

\section{Existence and uniqueness of the solution}\label{appendix existence}
	For all $n\geq n_0$ we have $D_n(f)\neq 0$ so there exists a unique solution {$x_j^{(n)}$, for $j=0,1,...,n-1$, of the linear problem
\begin{equation}\label{linear problem general}
\sum_{k=0}^{n-1} f_{j-k}x_k^{(n)}=y_j^{(n)}, \quad \textrm{for } j=0,1,...,n-1,
\end{equation}
for arbitrary complex numbers $y_j^{(n)}$, $j=0,1,\ldots n- 1$.
} We define the coefficients
	\begin{equation}\label{u v definition}
	u_j^{(n)}=\begin{cases}
	\sum_{k=0}^{n-1} f_{j-k+n} x_k^{(n)} , \quad &\textrm{for } j=0,1,2,...\\
	0, \quad &\textrm{for } j=-1,-2,...
	\end{cases}
	,\quad
	v_j^{(n)}=\begin{cases}
	0 , \quad &\textrm{for } j=0,1,2,...\\
	\sum_{k=0}^{n-1} f_{j-k} x_k^{(n)} , \quad &\textrm{for } j=-1,-2,...
	\end{cases}
	\end{equation}
	and the functions
	\begin{equation}
	U^{(n)}=\sum_{j=0}^{\infty}u_j^{(n)}z^j , \qquad V^{(n)}=\sum_{j=1}^{\infty}v_{-j}^{(n)}z^{-j}.
	\end{equation}
	The functions $U^{(n)}$ and $V^{(n)}$ are well defined, and therefore analytic, on the same annulus as $f(z)$, the one defined by \eqref{rho plus minus}. To see this pick some $z$ from the annulus. We have
	\begin{equation}
	\begin{split}
	\sum_{j=1}^{\infty}&|v_{-j}^{(n)}||z|^{-j}\leq\sum_{j=1}^{\infty} \sum_{k=0}^{n-1}|f_{-j-k}||x_k^{(n)}||z|^{-j}=\sum_{k=0}^{n-1}|x_k^{(n)}||z|^{k}\sum_{j=1}^{\infty}|f_{-j-k}||z|^{-j-k}\\ &\leq\Big(\sum_{j=-\infty}^{\infty}|f_j||z|^{j}\Big) \Big(  \sum_{k=0}^{n-1}|x_k^{(n)}||z|^{k} \Big)<\infty \;,
	\end{split}
	\end{equation}
	where the last inequality holds because Laurent series is absolutely convergent in the interior of its annulus. In an analogous way it is shown that $U^{(n)}$ is well defined.
	
	It follows from definition \eqref{u v definition} that the equation
	\begin{equation}
	\sum_{k=0}^{n-1}f_{j-k}x_k^{(n)}=y_j^{(n)}+u_{j-n}^{(n)}+v_j^{(n)},
	\end{equation}
	with $y_j^{(n)}$ defined {to be zero for $j<0$ and $j\geq n$}, holds for all $j\in\mathbb{Z}$. Multiplying the equation by $z^j$, with $z$ belonging to the annulus \eqref{rho plus minus}, and summing from $j=-\infty$ to $j=\infty$ it follows
	\begin{equation}\label{functional equation app}
	f(z)\xw(z)=\yw(z)+\uw(z) z^n+ \vw(z) \; ,
	\end{equation}
where
{
\begin{equation}\label{X Y definition app}
X^{(n)}(z)=\sum_{j=0}^{n-1}x_j^{(n)}z^j, \quad Y^{(n)}(z)=\sum_{j=0}^{n-1}y_j^{(n)}z^j .
\end{equation}
} Thus we have shown that{, for arbitrary polynomials $\yw$ of degree not greater than $n-1$, i.e. analytic functions with the properties
\begin{equation}
\big[\yw\big]_-=0,\quad \big[\yw z^{-n}\big]_+=0,
\end{equation}
that the functions $X^{(n)}, U^{(n)}, V^{(n)}$ are the solution of the functional equation \begin{equation}\label{equivalent problem app}
fX^{(n)}=Y^{(n)}+U^{(n)}z^{n}+V^{(n)}
\end{equation}
on the annulus and by construction they have the properties
\begin{equation}\label{properties app}
\big[X^{(n)}\big]_-=0,\quad \big[X^{(n)}z^{-n}\big]_+=0,\quad \big[U^{(n)}\big]_-=0, \quad \big[V^{(n)}\big]_+=0 \; .
\end{equation} The uniqueness of the solution of \eqref{linear problem general} implies the uniqueness of the solution of \eqref{equivalent problem app} under constraint \eqref{properties app}. }

\section{Wiener-Hopf procedure}\label{appendix solution}

\subsection{Wiener-Hopf equations}
We assume $\wn\geq0$. {The determinant \eqref{first} for $\wn <0$ can be obtained simply by transposing the Toeplitz matrix for the opposite sign of $\wn$ and making the integral transformation $\theta \to -\theta$.} From \eqref{equivalent problem app} it follows, separating the components,
\begin{equation}
a_+ z^\wn \xw-\big[\amm\yw\big]_+-\big[\amm\uw z^{n}\big]_+=\amm\vw+\big[\amm\yw\big]_-+\big[\amm\uw z^{n}\big]_- ,
\end{equation}
where $a_\pm$ have been defined in \eqref{apmdef}.
We now use the standard Wiener-Hopf argument \cite{TheTwo-DimensionalIsingModel}. Namely, the properties {\eqref{properties app}} imply that through it's Laurent series the left-hand side defines a function analytic on $|z|<\rho_+$,  while the right-hand side defines a function analytic on $|z|> \rho_-$, that goes to zero for $|z|\to\infty$. The two sides together define a function analytic on the whole plane and zero at infinity, thus, by Liouville's theorem, zero on the whole plane. It follows
\begin{align}
&\xw z^\wn=\apm \big(\big[\amm\yw\big]_++\big[\amm\uw z^{n}\big]_+\big) \label{w x 1},\\
&\vw=-\amp \big(\big[\amm\yw\big]_-+\big[\amm\uw z^{n}\big]_- \big) \label{w v}.
\end{align}

Similarly, denoting
\begin{equation}
\ucoeff{k}=\big(\apm\uw z^{-\wn}\big)_{-k} \; ,
\end{equation}
and multiplying {\eqref{equivalent problem app}} by $\apm z^{-(n+\wn)}$, we can make the separation
\begin{equation}
\begin{split}
\Big(&\apm\uw z^{-\wn}-\sum_{k=1}^{\wn}\ucoeff{k}z^{-k}\Big)+ \big[\apm\yw z^{-(n+\wn)}\big]_++\big[\apm\vw z^{-(n+\wn)}\big]_+\\
&=\amp\xw z^{-n}-\big[\apm\yw z^{-(n+\wn)}\big]_--\big[\apm\vw z^{-(n+\wn)}\big]_- -\sum_{k=1}^{\wn}\ucoeff{k}z^{-k} \; .
\end{split}
\end{equation}
It follows
\begin{align}
&\uw z^{-\wn}=-\app\big(\big[\apm\yw z ^{-(n+\wn)}\big]_++\big[\apm\vw z^{-(n+\wn)}\big]_+\big)+\app\sum_{k=1}^{\wn}\ucoeff{k} z^{-k} \label{w u},\\
&{X^{(n)}}z^{-n}=\amm\big(\big[\apm\yw z ^{-(n+\wn)}\big]_-+\big[\apm\vw z^{-(n+\wn)}\big]_-\big)+\amm\sum_{k=1}^{\wn}\ucoeff{k} z^{-k} \label{w x 2}.
\end{align}
This result is also valid for $\wn=0$ adopting the convention $\sum_{k=1}^{0}=0$.

The solution of the set of equations \eqref{w x 1}, \eqref{w v}, \eqref{w u} and \eqref{w x 2}, together with the requirement
\begin{equation}\label{requirement}
\big(\xw z^\wn\big)_j=0 \quad \textrm{for } {j=0,1,...,\wn-1,}
\end{equation}
that fixes the coefficients $\ucoeff{1},\ucoeff{2},...,\ucoeff{\wn}$, is the solution of {\eqref{equivalent problem app}} with the desired properties {\eqref{properties app}}.

\subsection{The solution}
{In this section we solve} {asymptotically} {the functional equation \eqref{equivalent problem app} with $\yw$ defined by \eqref{y definition} and \eqref{Y definition}. }For the set of equations \eqref{w x 1}, \eqref{w v}, \eqref{w u} and \eqref{w x 2} a solution in the closed form might not exist so we follow the standard approach \cite{Wu1966,TheTwo-DimensionalIsingModel} and we look for the solution by making an assumption on the function $U^{(n)}$ and then checking whether the final solution we obtain is consistent with this assumption.

We assume that
\begin{equation}\label{assumption}
U^{(n)}z^{-\wn}-\app\sum_{k=1}^{\wn}\ucoeff{k}z^{-k}=O(1) \quad \textrm{ uniformly {in} $z$, on $\rho\leq|z|\leq \rho^{-1}$, for all $\rho$ defined by \eqref{rho definition}}.
\end{equation}
The second term in \eqref{w v} is equal to
\begin{equation}\label{step v a}
\big[\amm\uw z^{n}\big]_ -=
\Big[\amm z^{n+\wn}\Big(U^{(n)}z^{-\wn}
-\app\sum_{k=1}^{\wn}\ucoeff{k}z^{-k}\Big)\Big]_-
+ \left[ \app\amm \sum_{k=1}^{\wn}\ucoeff{k}z^{n+\wn-k} \right]_- \, .
\end{equation}
Applying Lemma \ref{lemma plus minus components} on the first term in \eqref{step v a} gives
\begin{equation}\label{step u a}
\big[\amm\uw z^{n}\big]_-=
\left[ \app\amm \sum_{k=1}^{\wn}\ucoeff{k}z^{n+\wn-k} \right]_- +O(\rho^n) \; ,
\end{equation}
where $O(\rho^n)$ holds on $\rho\leq|z|\leq \rho^{-1}$, uniformly in $z$, for all $\rho$ satisfying \eqref{rho definition}. From now on it is always the case and we don't write every time that $O$ holds uniformly in $z$ on $\rho\leq|z|\leq \rho^{-1}$, for all $\rho$ satisfying \eqref{rho definition}. We can thus write \eqref{w v} as
\begin{equation}\label{v before u}
\vw(z)=-\amp(z)\big[\amm\yw\big]_-(z)-\amp(z)\sum_{k=1}^{\wn}\ucoeff{k} \big[\app\amm z^{n+\wn-k}\big]_-(z)+O(\rho^n).
\end{equation}

We use \eqref{Y explicitly} to rewrite the first term on the RHS of \eqref{v before u} as
\begin{equation}\label{step y a}
\big[\amm\yw\big]_-=\etm{(n-1)}\bigg[\frac{\amm z^n}{z-\et{}}\bigg]_-^{\ts}-\et{}\bigg[\frac{\amm}{z-\et{}}\bigg]_-^{\ts} \; .
\end{equation}
The first term here is $O(\rho^n)$ by Lemma \ref{lemma plus minus components}. Applying Lemma \ref{lemma components t} to the second term gives
\begin{equation}\label{step y a1}
\big[\amm\yw\big]_-(z)=-\et{}\frac{\amm(z)-\amm(\et{})}{z-\et{}}+O(\rho^n) \quad \textrm{ for $z\neq\et{}$, $\rho\leq|z|\leq \rho^{-1}$}.
\end{equation}
The value at $z=\et{}$ is obtained by continuity and from now on we omit writing $z\neq\et{}$, $\rho\leq|z|\leq \rho^{-1}$ every time. It follows
\begin{equation}\label{obtained V}
\vw(z)=\et{}\amp(z)\frac{\amm(z)-\amm(\et{})}{z-\et{}}-\amp(z)\sum_{k=1}^{\wn}\ucoeff{k}\big[\amm\app z^{n+\wn-k}\big]_-(z)+O(\rho^n) \; .
\end{equation}

This expression can be used in \eqref{w u} to find $\uw$. Before we do so, we use \eqref{Y explicitly} again to get for the first term on the RHS of \eqref{w u}:
\begin{equation}
\big[\apm\yw z^{-(n+\wn)}\big]_+=\etm{(n-1)}\bigg[\frac{\apm z^{-\wn}}{z-\et{}}\bigg]_+^{\tl}-\et{}\bigg[\frac{\apm z^{-(n+\wn)}}{z-\et{}}\bigg]_+^{\tl} \; .
\end{equation}
The second term is $O(\rho^n)$ by Lemma \ref{lemma plus minus components}. Applying Lemma \ref{lemma components t} to the first term we get
\begin{equation}
\big[\apm\yw z^{-(n+\wn)}\big]_+(z)=\etm{(n-1)}\frac{[\apm z^{-\wn}]_+(z)-[\apm z^{-\wn}]_+(\et{})}{z-\et{}}+O(\rho^n) \; .
\end{equation}
We can now substitute \eqref{obtained V} in \eqref{w u} and apply Lemma \ref{lemma plus minus components} to the second term on the RHS of \eqref{w u} to get
\begin{equation}
\big[\apm\vw z^{-(n+\wn)}\big]_+(z)=-\sum_{k=1}^{\wn} \ucoeff{k} \big[\amp\apm z^{-(n+\wn)}\big[\amm\app z^{n+\wn-k}\big]_-\big]_+(z) +O(\rho^n) \; .
\end{equation}
Collecting everything it follows from \eqref{w u}
\begin{equation}\label{obtained u}
\begin{split}
\uw z^{-\wn}=&-\etm{(n-1)}\app(z)\frac{[\apm z^{-\wn}]_+(z)-[\apm z^{-\wn}]_+(\et{})}{z-\et{}}+\\
& +\app(z)\sum_{k=1}^{\wn}\ucoeff{k} \Big(z^{-k} +\big[\amp\apm z^{-(n+\wn)}\big[\amm\app z^{n+\wn-k}\big]_-\big]_+(z)\Big)+O(\rho^n).
\end{split}
\end{equation}

The coefficients $\ucoeff{1},\ucoeff{2},...,\ucoeff{\wn}$ remain to be determined. However, if we assume that, for sufficiently small $\rho$, they satisfy
\begin{equation}\label{if satisfy}
\ucoeff{k}O(\rho^{2n})=O(1) , \quad \textrm{for }k=1,2,...,\wn,
\end{equation}
then, taking a $\rho_1$ such that $\rho_-<\rho_1<\rho<1<\rho^{-1}<\rho_1^{-1}<\rho_+$, the last term in \eqref{obtained u} is, by Lemma \ref{lemma plus minus components},
\begin{equation}
\app(z)\sum_{k=1}^{\wn}\ucoeff{k}\big[\amp\apm z^{-(n+\wn)}\big[\amm\app z^{n+\wn-k}\big]_-\big]_+(z)=\app(z)\sum_{k=1}^{\wn}\ucoeff{k}O(\rho_1^{2n})=O\big( (\rho_1/\rho)^2\big), \quad \textrm{on }\rho\leq |z|\leq\rho^{-1} .
\end{equation}
It follows
\begin{equation}\label{obtained u 1}
\uw z^{-\wn}=-\etm{(n-1)}\app(z)\frac{[\apm z^{-\wn}]_+(z)-[\apm z^{-\wn}]_+(\et{})}{z-\et{}} +\app(z)\sum_{k=1}^{\wn}\ucoeff{k} z^{-k} +O(\sigma^n),
\end{equation}
where $\sigma=\max\{(\rho_1/\rho)^2,\rho\} $. Then \eqref{obtained u 1} is consistent with the starting assumption \eqref{assumption}, while assumption \eqref{if satisfy} will be checked below for its consistency.

Finally, $\xw$ is computed using \eqref{w x 1}. The first term in \eqref{w x 1} is found from \eqref{Y explicitly} and \eqref{step y a1}, using
\begin{equation}
\big[\amm\yw\big]_+=\amm\yw-\big[\amm\yw\big]_- \; .
\end{equation}
We get
\begin{equation}\label{ammy}
\big[\amm\yw\big]_+(z)=
\etm{(n-1)}\frac{\amm(z)z^n-\amm(\et{})\et{n}}{z-\et{}}+O(\rho^n).
\end{equation}
The second term in \eqref{w x 1} is found from \eqref{step u a} and \eqref{obtained u 1},
\begin{eqnarray}
\big[\amm\uw z^{n}\big]_+ & = &\amm\uw z^{n}-\big[\amm\uw z^{n}\big]_-
\nonumber \\
& = & -\etm{(n-1)}\amm(z) \app(z) z^{n+\wn} \frac{\apm(z) z^{-\wn}- \apm (\et{}) \etm{\wn}}{z-\et{}} 
\nonumber \\
&&+\etm{(n-1)}\amm(z)\app(z)z^{n+\wn}\sum_{k=0}^{\wn-1}(\apm)_k\frac{z^{k-\wn}-\et{(k-\wn)}}{z-\et{}}
\nonumber \\
&& + \left[ \app\amm \sum_{k=1}^{\wn}\ucoeff{k}z^{n+\wn-k} \right]_+ (z) +O(\sigma^n) \; ,
\label{ammUz}
\end{eqnarray}
where we used
\begin{equation}
\big[\apm z^{-\wn}\big]_+=\apm z^{-\wn}-\sum_{k=0}^{\wn-1} (\apm)_kz^{k-\wn} \; .
\end{equation}

Now, summing \eqref{ammy} and \eqref{ammUz} in \eqref{w x 1} gives
\begin{equation}\label{obtained x 1} 
\begin{split}
\xw(z)z^\wn=&\etm{(n+\wn-1)}\apm(z)\apm(\et{}) \frac{\app(z)\amm(z)z^{n+\wn}-\app(\et{})\amm(\et{})\et{(n+\wn)}}{z-\et{}}\\
&+\etm{(n-1)}\amm(z)z^{n+\wn}\sum_{k=0}^{\wn-1}(\apm)_k\frac{z^{k-\wn}-\et{(k-\wn)}}{z-\et{}}+\apm(z)\sum_{k=1}^{\wn}\ucoeff{k}\big[\amm\app z^{n+\wn-k}\big]_+(z)+O(\sigma^n) \; .
\end{split}
\end{equation}

It remains to determine the coefficients $\ucoeff{1},\ucoeff{2},...,\ucoeff{\wn}$ from requirement \eqref{requirement} and to see whether \eqref{if satisfy} is satisfied. We compute the coefficients $(\xw z^{-\wn})_j$ by \eqref{Laurent coefficients}, integrating at $|w|=\rho$. All the terms in \eqref{obtained x 1} containing the factor $z^n$ result in $O(\rho^n)$ corrections, while
\begin{equation}
\frac{1}{2\pi\imath}\oint_{|w|=\rho}\frac{\apm(w)}{w-\et{}}\frac{dw}{w^{j+1}}=-\etm{}\sum_{k=0}^{j} \big(\apm\big)_k \etm{(j-k)}.
\end{equation}
It follows
\begin{equation}
\big(\xw z^{\wn}\big)_j=\sum_{k=1}^{\wn} \ucoeff{k}\Big(\apm\big[c z^{n+\wn-k}\big]_+\Big)_j+\amm(\et{})\sum_{k=0}^{j} \big(\apm\big)_k \etm{(j-k)}+O(\sigma^n),
\end{equation}
where $c=\app \amm$.

Thus if the coefficients $\ucoeff{1},\ucoeff{2},...,\ucoeff{\wn}$ satisfy
\begin{equation}\label{step equations}
0=\sum_{k=1}^{\wn} \ucoeff{k}\Big(\apm\big[c z^{n+\wn-k}\big]_+\Big)_j+\amm(\et{})\sum_{k=0}^{j} \big(\apm\big)_k \etm{(j-k)},\quad \textrm{for }j=0,1,...,\wn-1,
\end{equation}
then
\begin{equation}
\big(\xw z^\wn\big)_j=O(\sigma^n), \quad \textrm{for }j=0,1,...,\wn-1.
\end{equation}
Using
\begin{equation}
\Big(\apm\big[c z^{n+\wn-k}\big]_+\Big)_j=\sum_{m=0}^{j}\big(\apm\big)_{m} c_{j-m-n-\wn+k}
\end{equation}
it's easy to see that \eqref{step equations} is equivalent to
\begin{equation}\label{step equations 1}
\sum_{k=1}^{\wn}\ucoeff{k}c_{j-n-\wn+k}=-\amm(\et{})\etm{j}, \quad \textrm{for }j=0,1,...,\wn-1.
\end{equation}
The set of equations \eqref{step equations 1} is solved by Cramer's rule. The solution is
\begin{equation}
\ucoeff{j}=-\amm(\et{})\etm{(\wn-1)}\frac{\tilde{\Delta}_{\wn,n}(j)}{\Delta_{\wn,n}}
\end{equation}
where $\Delta_{\wn,n}$ and $\tilde{\Delta}_{\wn,n}(j)$ are defined by \eqref{Delta} and \eqref{Delta tilde l} respectively. We see that the condition \eqref{assumption} of Theorem \ref{tm non-zero winding} ensures that the assumption \eqref{if satisfy} is satisfied.

The solution of equations \eqref{w x 1}, \eqref{w v} and \eqref{w u} we have found is consistent with assumptions \eqref{assumption} and \eqref{if satisfy}, that we have made to find it, up to $O(\sigma^n)$ terms. On the basis of this solution we construct the functions $\xw_1,\uw_1$ and $\vw_1$ discussed in sections \ref{sec zero winding} and \ref{sec non zero winding}.

\section{Remarks on rigor} \label{Appendix rigor}

To prove theorems \ref{tm zero winding} and \ref{tm non-zero winding} we have used the intuitive property that a small perturbation to the functional equation \eqref{equivalent problem app} results only in a small perturbation to the solutions. Here we show it rigorously for the case of a zero winding number of the symbol, and thus make Theorem \ref{tm zero winding} rigorous.

The result we need is given by Lemma \ref{lemma rigor}, that we are going to state and prove. The proof of the lemma uses a similar procedure to the one used in the proof of Szeg\H{o} theorem in chapter X of \cite{TheTwo-DimensionalIsingModel}, based on the Wiener-Hopf equations. Let us note that the complication with non-zero winding number is the presence of the third term in \eqref{w u} and we omit this case.

Before introducing the lemma let us introduce two norms for functions analytic on an annulus around the unit circle. The first one is the supremum norm on the unit circle. Let $g$ be analytic on the annulus $\rho_-<|z|<\rho_+$ that includes the unit circle. We define
\begin{equation}
\norm{g}=\sup_{z\in\mathbb{C}:|z|=1} |g(z)| .
\end{equation}
The second norm is defined as the sum of the absolute values of Laurent series coefficients of $g$, which is well defined since the Laurent series is absolutely convergent in the interior of its annulus. If $g(z)=\sum_{j=-\infty}^{\infty}g_j z^j$, we denote
\begin{equation}
\normb{g}=\sum_{j=-\infty}^{\infty}|g_j|.
\end{equation}
We have clearly
\begin{equation}
\norm{g}\leq\normb{g}.
\end{equation}

Let us discuss the properties of the norms related to the components \eqref{components definition}. From their integral representation \eqref{plus minus integral representation} we have
\begin{equation}\label{components supremum norm}
\norm{[g]_-}\leq \frac{\rho_1}{1-\rho_1}\sup_{z:|z|=\rho_1}\{|g(z)|\}, \qquad \norm{[g]_+}\leq \frac{\rho_2}{\rho_2-1}\sup_{z:|z|=\rho_2}|g(z)|,
\end{equation}
for any choice $\rho_1\in(\rho_-,1),\rho_2\in(1,\rho_+)$. On the other hand, the second norm clearly satisfies
\begin{equation}
\normb{[g]_-}\leq \normb{g}, \qquad \normb{[g]_+}\leq \normb{g}.
\end{equation}

If $g$ and $h$ are two functions analytic on an annulus around the unit circle we have
\begin{equation}
\norm{gh}\leq \norm{g}\norm{h}, \qquad \normb{gh}\leq \normb{g}\normb{h},
\end{equation}
where the first inequality is trivial and the second is proven easily using the absolute convergence of the Laurent series inside the annulus. Finally, for a sequence of functions $(g^{(n)})_{n\in\mathbb{N}}$ analytic on an annulus around the unit circle and a sequence of positive numbers $(s_n)_{n\in\mathbb{N}}$ we have clearly that
\begin{equation}
g^{(n)}(z)=O(s_n) \quad \textrm{uniformly in $z$ on the unit circle $|z|=1$}  \quad \textrm{if and only if} \quad \norm{g^{(n)}}=O(s_n).
\end{equation}

\begin{lemma}\label{lemma rigor}
	Let $a$ be non-zero analytic function on an annulus around the unit circle, defined by \eqref{rho plus minus}, and with a zero winding number. Let $\yw_j$, for $j=1,2$ and $n\in\mathbb{N}$, be polynomials of degree not greater than $n-1$, i.e. analytic functions with
	\begin{equation}
	\big[Y_j^{(n)}\big]_-=0,\quad \big[Y_j^{(n)}z^{-n}\big]_+=0,
	\end{equation}
	and let $X_j^{(n)},U_j^{(n)},V_j^{(n)}$, for $j=1,2$, be analytic solutions of the functional equation
	\begin{equation}
	aX_j^{(n)}=Y_j^{(n)}+U_j^{(n)}z^n+V_j^{(n)}
	\end{equation}
	on the annulus, such that they satisfy the properties 
	\begin{equation}\label{properties App}
	\big[X_j^{(n)}\big]_-=0,\quad \big[X_j^{(n)}z^{-n}\big]_+=0,\quad \big[U_j^{(n)}\big]_-=0, \quad \big[V_j^{(n)}\big]_+=0 \; .
	\end{equation}
	If
	\begin{equation}
	Y_1^{(n)}(z)-Y_2^{(n)}(z)=O(s_n) \quad \textrm{uniformly in $z$ on the unit circle $|z|=1$}
	\end{equation}
	then
	\begin{equation}\label{x lemma result}
	X_1^{(n)}(z)-X_2^{(n)}(z)=O(ns_n) \quad \textrm{uniformly in $z$ on the unit circle $|z|=1$.}
	\end{equation}
	
\end{lemma}

\begin{proof}
 Defining
	\begin{equation}
	Y^{(n)}=Y_1^{(n)}-Y_2^{(n)}, \quad X^{(n)}=X_1^{(n)}-X_2^{(n)}, \quad U^{(n)}=U_1^{(n)}-U_2^{(n)}, \quad V^{(n)}=V_1^{(n)}-V_2^{(n)},
	\end{equation}
	proving the lemma becomes equivalent to showing that if $\norm{Y^{(n)}}=O(s_n)$ then the solution $X^{(n)}$ of the problem \eqref{equivalent problem app} (for $\wn=0$) and \eqref{properties app} satisfies $\norm{X^{(n)}}=O(ns_n)$.
	
	The lines of the proof are the following. First we use the Wiener-Hopf equations \eqref{w v} and \eqref{w u} to show that $\norm{U^{(n)}}=O(ns_n)$ and $\norm{V^{(n)}}=O(ns_n)$. Then we recognize that directly from \eqref{equivalent problem app} and the properties of the norm it follows
	\begin{equation}\label{ineq y u v}
	\norm{X^{n}}\leq \norm{a^{-1}}\Big(\norm{\yw}+\norm{\uw}+\norm{\vw}\Big),
	\end{equation}
	and conclude that since the right hand side of the inequality is $O(ns_n)$, so is the left side.
	
	We now work out the details. Note first that the Laurent series coefficients $y_j^{(n)}$ of the function $Y^{(n)}(z)=\sum_{j=0}^{n-1}y_j^{(n)}z^j $ satisfy $|y_j^{(n)} |\leq\norm{Y^{(n)}}$, which can be seen easily from the integral representation of the coefficients. It follows
	\begin{equation}\label{two norms Y}
	\normb{Y^{(n)}}\leq n\norm{Y^{(n)}}.
	\end{equation}
	From the Wiener-Hopf equation \eqref{w v} we get
	\begin{equation}\label{v bound}
\norm{\vw}=\norm{\amp} \norm{\big[\amm\yw\big]_-} +\norm{\amp} \norm{\big[\amm\uw z^{n}\big]_-}.
	\end{equation}
	To bound the first term we notice
	\begin{equation}
	 \norm{\big[\amm\yw\big]_-}\leq \normb{\big[\amm\yw\big]_-} \leq \normb{\amm\yw}\leq \normb{\amm}\normb{\yw} \leq \normb{\amm} n \norm{\yw},
	\end{equation}
	where all the inequalities except the last follow simply from the discussed properties of the norms, and the last one follows from \eqref{two norms Y}. To bound the second term in \eqref{v bound} we use \eqref{components supremum norm} with $\rho_1=\rho$ for some $\rho$ defined by \eqref{rho definition}. We have
	\begin{equation}
	\norm{\big[\amm\uw z^{n}\big]_-}\leq \frac{\rho^{n+1}}{1-\rho}	\sup_{z:|z|=\rho}\{|\amm (z)  U^{(n)}(z)|\}.
	\end{equation}
	Since the Laurent series coefficients of $\uw$ define a function analytic inside the whole circle $|z|<\rho^{-1}$ we can apply the maximum modulus principle to conclude
	\begin{equation}
	\sup_{z:|z|=\rho}\{|U^{(n)}(z)|\}\leq \sup_{z:|z|=1}\{|U^{(n)}(z)|\}=\norm{U^{(n)}}.
	\end{equation}
	It follows
	\begin{equation}
	\norm{\big[\amm\uw z^{n}\big]_-}\leq \frac{\rho^{n+1}}{1-\rho}	\sup_{z:|z|=\rho}\{|\amm(z) |  \} \norm{\uw}.
	\end{equation}
	
	We conclude that there are positive constants $\lambda_1$ and $\lambda_2$ (independent of $n$) such that
	\begin{equation}
	\norm{\vw}\leq \lambda_1 n \norm{\yw }+\lambda_2 \rho^{n}\norm{\uw}.
	\end{equation}
	Using the same methods we conclude from the Wiener-Hopf equation \eqref{w u} that there are positive constants $\lambda_3$ and $\lambda_4$ such that
	\begin{equation}\label{ineq 2}
	\norm{\uw}\leq \lambda_3 n \norm{\yw }+\lambda_4 \rho^{n}\norm{\vw}.
	\end{equation}
	
	We have obtained a system of two inequalities. Inserting the second into the first, and rearranging the terms, we get
	\begin{equation}
	(1-\lambda_2 \lambda_4 \rho^{2n})\norm{\vw}\leq (\lambda_1+\lambda_2\lambda_3 \rho^n )n\norm{\yw}.
	\end{equation}
	For sufficiently large $n$ the factor on the right is positive and greater than, say, $1/2$, so we conclude
	\begin{equation}
	\norm{\vw}=O(n s_n).
	\end{equation}
	Using \eqref{ineq 2} again we get also
	\begin{equation}
	\norm{\uw}=O(n s_n).
	\end{equation}
	Finally, from \eqref{ineq y u v} we get then $\norm{\xw}=O(n s_n)$, which completes the proof.
\end{proof}

In applying Lemma \ref{lemma rigor} for functions $\yw_2$ and $\yw$ in section \ref{sec zero winding} we have $s_n=\rho^n$, for any $\rho$ defined by \eqref{rho definition}. Note that since $\rho$ can always be made smaller, the factor $n$ in $O(ns_n)$ in \eqref{x lemma result} is irrelevant.

\end{appendices}

\bibliographystyle{hunsrt}
\bibliography{bibliography}

\begin{thebibliography}{10}

\bibitem{Bottcher1999}
Albrecht B{\"o}ttcher and Bernd Silbermann.
\newblock {\em Introduction to Large Truncated Toeplitz Matrices}.
\newblock Springer-Verlag New York, 1999.

\bibitem{Rao1991}
Murali Rao and Henrik Stetkaer.
\newblock {\em Complex Analysis: An Invitation}.
\newblock World Scientific Publishing Co., 1991.

\bibitem{Szego}
G{\'a}bor Szeg{\"o}.
\newblock On certain hermitian forms associated with the fourier series of a
  positive function.
\newblock {\em Festschrift Marcel Riesz}, 1952.

\bibitem{DeiftItsKrasovsky}
Percy Deift, Alexander Its, and Igor Krasovsky.
\newblock Toeplitz matrices and toeplitz determinants under the impetus of the
  ising model: Some history and some recent results.
\newblock {\em Communications on Pure and Applied Mathematics},
  66(9):1360--1438, 2013,
  https://onlinelibrary.wiley.com/doi/pdf/10.1002/cpa.21467.

\bibitem{Simon_orthogonal}
Barry Simon.
\newblock {\em Orthogonal polynomials on the unit circle}.
\newblock Colloquium Publications, 2004.

\bibitem{Bottcher_book}
Albrecht B{\"o}ttcher and Bernd Silbermann.
\newblock {\em Analysis of Toeplitz Operators}.
\newblock Springer-Verlag Berlin Heidelberg, 2006.

\bibitem{Hartwig1969}
Robert~E. Hartwig and Michael~E. Fisher.
\newblock Asymptotic behavior of toeplitz matrices and determinants.
\newblock {\em Archive for Rational Mechanics and Analysis}, 32(3):190--225,
  Jan 1969.

\bibitem{FisherHartwig}
Michael~E. Fisher and Robert~E. Hartwig.
\newblock {\em Toeplitz Determinants: Some Applications, Theorems, and
  Conjectures}, pages 333--353.
\newblock John Wiley and Sons, Ltd, 1969.

\bibitem{BottcherSilbermann}
Albrecht B{\"o}ttcher and Bernd Silbermann.
\newblock Notes on the asymptotic behavior of block toeplitz matrices and
  determinants.
\newblock {\em Mathematische Nachrichten}, 98(1):183--210, 1980,
  https://onlinelibrary.wiley.com/doi/pdf/10.1002/mana.19800980116.

\bibitem{BottcherWidom}
Albrecht B{\"o}ttcher and Harold Widom.
\newblock Szeg{\"o} via jacobi.
\newblock {\em Linear Algebra and its Applications}, 419(2):656 -- 667, 2006.

\bibitem{Krasovsky_aspects}
Igor Krasovsky.
\newblock Aspects of toeplitz determinants.
\newblock In Daniel Lenz, Florian Sobieczky, and Wolfgang Woess, editors, {\em
  Random Walks, Boundaries and Spectra}, pages 305--324, Basel, 2011. Springer
  Basel.

\bibitem{DongPRE}
Jian-Jun Dong, Zhen-Yu Zheng, and Peng Li.
\newblock Rigorous proof for the nonlocal correlation function in the
  transverse ising model with ring frustration.
\newblock {\em Phys. Rev. E}, 97:012133, Jan 2018.

\bibitem{Dong_2016}
Jian-Jun Dong, Peng Li, and Qi-Hui Chen.
\newblock The a-cycle problem for transverse ising ring.
\newblock {\em Journal of Statistical Mechanics: Theory and Experiment},
  2016(11):113102, nov 2016.

\bibitem{mari2019frustration}
Vanja Mari{\'{c}}, Salvatore~Marco Giampaolo, Domagoj Kui{\'{c}}, and Fabio
  Franchini.
\newblock The frustration of being odd: how boundary conditions can destroy
  local order.
\newblock {\em New Journal of Physics}, 22(8):083024, aug 2020.

\bibitem{mari2020frustration}
Vanja Mari{\'{c}}, Salvatore~Marco Giampaolo, and Fabio Franchini.
\newblock Quantum phase transition induced by topological frustration.
\newblock {\em Communications Physics}, 3(1):220, Dec 2020.

\bibitem{Giampaolo:2018imc}
Salvatore~Marco Giampaolo, Fl{\'a}via~Braga Ramos, and Fabio Franchini.
\newblock {The Frustration of being Odd: Universal Area Law violation in local
  systems}.
\newblock {\em J. Phys. Comm.}, 3(8):081001, 2019, 1807.07055.

\bibitem{TheTwo-DimensionalIsingModel}
Barry McCoy and Tai~Tsun Wu.
\newblock {\em The Two-Dimensional Ising Model}.
\newblock Harvard University Press, 1973.

\bibitem{Wu1966}
Tai~Tsun Wu.
\newblock Theory of toeplitz determinants and the spin correlations of the
  two-dimensional ising model. i.
\newblock {\em Phys. Rev.}, 149:380--401, Sep 1966.

\bibitem{LIEB1961407}
Elliott Lieb, Theodore Schultz, and Daniel Mattis.
\newblock Two soluble models of an antiferromagnetic chain.
\newblock {\em Annals of Physics}, 16(3):407 -- 466, 1961.

\bibitem{McCoy1968}
Barry~M. McCoy.
\newblock Spin correlation functions of the $x\ensuremath{-}y$ model.
\newblock {\em Phys. Rev.}, 173:531--541, Sep 1968.

\bibitem{McCoy2}
Eytan Barouch and Barry~M. McCoy.
\newblock Statistical mechanics of the $xy$ model. ii. spin-correlation
  functions.
\newblock {\em Phys. Rev. A}, 3:786--804, Feb 1971.

\bibitem{Franchini:2016cxs}
Fabio Franchini.
\newblock {\em {An introduction to integrable techniques for one-dimensional
  quantum systems}}.
\newblock Springer, 2017.

\bibitem{CampostriniRings}
Massimo Campostrini, Andrea Pelissetto, and Ettore Vicari.
\newblock Quantum transitions driven by one-bond defects in quantum ising
  rings.
\newblock {\em Phys. Rev. E}, 91:042123, Apr 2015.

\bibitem{Campostrini_2015}
Massimo Campostrini, Andrea Pelissetto, and Ettore Vicari.
\newblock Quantum ising chains with boundary fields.
\newblock {\em Journal of Statistical Mechanics: Theory and Experiment},
  2015(11):P11015, nov 2015.

\end{thebibliography}

\end{document}